\documentclass[final,3p,mathptmx]{elsarticle}

\usepackage{lineno}
\usepackage{hyperref}
\usepackage{subcaption,wrapfig}
\usepackage{amsmath}
\usepackage{amssymb}
\usepackage{amsthm}
\usepackage{enumerate}
\usepackage[usenames,dvipsnames]{color}
\usepackage[T1]{fontenc}
\usepackage[utf8]{inputenc}
\usepackage{braket,mleftright}
\modulolinenumbers[5]

\newtheorem{theorem}{Theorem}

\journal{arXiv}

\begin{document}

\begin{frontmatter}

\title{A Physical Perspective on Control Points and Polar Forms: B\'{e}zier Curves, Angular Momentum and Harmonic Oscillators}

\author{Márton Vaitkus}\ead{vaitkus@iit.bme.hu}
\address{Budapest University of Technology and Economics}

\begin{abstract}
Bernstein polynomials and B\'{e}zier curves play an important role in computer-aided geometric design and numerical analysis, and their study relates to mathematical fields such as abstract algebra, algebraic geometry and probability theory. We describe a theoretical framework that incorporates the different aspects of the Bernstein-B\'{e}zier theory, based on concepts from theoretical physics. We relate B\'{e}zier curves to the theory of angular momentum in both classical and quantum mechanics, and describe physical analogues of various properties of B\'{e}zier curves -- such as their connection with polar forms -- in the context of quantum spin systems. This previously unexplored relationship between geometric design and theoretical physics is established using the mathematical theory of Hamiltonian mechanics and geometric quantization. An alternative description of spin systems in terms of harmonic oscillators serves as a physical analogue of P\'{o}lya's urn models for B\'{e}zier curves. We relate harmonic oscillators to Poisson curves and the analytical blossom as well. We present an overview of the relevant mathematical and physical concepts, and discuss opportunities for further research.
\end{abstract}


\end{frontmatter}

%\linenumbers

\section{Introduction}\label{sec:1_intro}
The Bernstein polynomial basis plays a central role in computer graphics and computer-aided geometric design (CAGD) \cite{farin2002curves,gallier2000curves}. The corresponding B\'{e}zier curves, surfaces and volumes are specified by a set of control points, which can be interacted with in a intuitive way, while also possessing superior numerical stability \cite{farouki2012bernstein} and efficient algorithms for evaluation, differentiation, subdivision and integration \cite{goldman2002pyramid}. Other widely used free-form curve and surface representations, such as B-Splines and Catmull-Clark/Doo-Sabin subdivision surfaces can also be decomposed into, or approximated with B\'{e}zier segments/patches. Use of the Bernstein polynomials is increasingly widespread in numerical analysis as well, in particular as shape functions for finite element methods \cite{borden2011isogeometric}. The study of Bernstein-B\'{e}zier representations already makes use of a wide range of mathematical fields including multi-linear algebra \cite{ramshaw2001paired}, algebraic geometry \cite{goldman2003topics} and probability theory \cite{goldman1985polya}. In this work -- first in a planned series -- we aim to introduce a novel theoretical framework, based on concepts from mathematical physics, which incorporates the major aspects of the Bernstein-B\'{e}zier theory. 

The key message we want to communicate is that B\'{e}zier curves and Bernstein polynomials have a close connection to both classical and quantum physics. In particular, Bernstein-B\'{e}zier theory for curves is related to the mechanics of 3-dimensional rotational motion, classical \emph{angular momentum} and quantum \emph{spin}. Remarkably, the important aspects of the Bernstein-B\'{e}zier theory, such as control points, polar forms, binomial probability distributions all make an appearance within the context of physics. The urn model analogy for B\'{e}zier curves \cite{goldman1985polya}, and the closely related Poisson curves \cite{morin2000subdivision,morin2002analytic} also have physical interpretations in terms of harmonic oscillators. The known correspondences are listed in \autoref{tab:list}.

\begin{table}[h]
	\centering
	\begin{tabular}[c]{|c|c|}
		\hline
		CAGD & Physics \\\hline\hline
		B\'{e}zier curves & Spin systems \\\hline
		Degree-$d$ & Spin-$\frac{d}{2}$ \\\hline
		Control points & Quantum eigenstates \\\hline
		Blending functions & Quantum coherent states \\\hline
		Polar forms & Quantum spin addition \\\hline
		Domain/Newton polytope & Moment map \\\hline
		Urn analogy & Oscillator analogy \\\hline
		Poisson curves & Harmonic oscillators \\ 
		\hline
	\end{tabular}
	\caption{Correspondences between CAGD and mathematical  physics}
	\label{tab:list}
\end{table}

The described relationship between B\'{e}zier curves and physics is more than a formal analogy; rather, it is based on rigorous mathematical arguments involving symplectic geometry and Lie group theory, as will be described in detail in an expository companion paper \cite{vaitkus2018physics}. These areas of mathematics (to the author's knowledge) have not been employed before in the context of Bernstein-B\'{e}zier representations. Introducing CAGD researchers to these mathematical tools is one of our major goals.
 
The presented physical interpretation also applies to B\'{e}zier surfaces and volumes (tensor products, simplices, and their toric generalizations \cite{krasauskas2002toric} as well), and can be described most generally using the language of group representation theory -- these will be the topic of planned follow-up papers \cite{vaitkus2018surface,vaitkus2018repthy}. We note that concepts and ideas originating from theoretical physics have come to play a major role in many mathematical fields, including differential and algebraic geometry, and topology \cite{atiyah2010geometry}. We hope that the presented physical interpretations can be similarly useful in the study and development of control point schemes in geometric design.

This paper is structured as follows. First, in \autoref{sec:motivation} we give a brief summary of the main aspects of the theory of B\'{e}zier curves that we aim to explain in a unified framework. Then, in \autoref{sec:rotational} we describe some basic facts about the classical and quantum physics of rotational motion, pointing out similarities with the mathematics of B\'{e}zier curves. In \autoref{sec:spinfacts}, we delve deeper into the quantum mechanics of spin, revealing further connections with the theory of polar forms and binomial distributions. The established correspondences are summarized in \autoref{sec:summary}. In \autoref{sec:bezier2physics}, \autoref{sec:classicalbasics} and \autoref{sec:bezierquantum} we sketch a mathematical explanation for these connections, using the theory of Hamiltonian mechanics and geometric quantization. In \autoref{sec:bezier_osc} we present an equivalent description of spin systems in terms of harmonic oscillators, as a physical analogue of P\'{o}lya's urn models. In \autoref{sec:poisson_osc} Poisson curves and the analytical blossom are related to harmonic oscillators. A survey in \autoref{sec:conc} of the many possible avenues for future research concludes the paper.

\section{Motivation and Background}\label{sec:motivation}
We are concerned with \emph{B\'{e}zier curves} -- vectors of degree-$d$ polynomials, that are of the form:
\begin{align}
\mathbf{P}:\ \left[ 0,1 \right]  & \longrightarrow \mathbb{R}^{n}\\
t &\longmapsto \mathbf{P}(t) = \sum_{i = 0}^{d}\mathbf{P}_{i}\beta_{i}^{d}(t).
\end{align}
The functions
\begin{align}
\beta_{i}^{d}(t) = {d \choose i}(1-t)^{d-i}t^{i} &\ \  i = 0, \ldots d
\end{align}
are called the \emph{Bernstein polynomials}, and constitute a basis for the $(d+1)$-dimensional space of degree-$d$ polynomials over $\mathbb{R}$. These polynomials have the remarkable properties
\begin{align}
\beta_{i}^{d}(t) &\geq 0;\qquad i = 0, \ldots d;\ t \in \left[0,1\right] \tag{Positivity}  \\
\sum_{i = 0}^{d}\beta_{i}(t) &= 1;\qquad t \in \mathbb{R} \tag{Partition of Unity}
\end{align}

We give an overview of some important aspects of Bernstein-B\'{e}zier representations -- for a thorough survey see \cite{farouki2012bernstein}.
\subsection{Control Points -- Geometric Design}\label{sec:motivation:bezier}
B\'{e}zier curves are useful for many applications due to the fact that in a Bernstein basis, the coefficients $\mathbf{P}_{i} \in \mathbb{R}^{n}$ have a clear geometric interpretation as \emph{control points}. 

These points form the \emph{control polygon} for the curve, as shown in \autoref{fig:beziercurve}. Since the Bernstein bases are pointwise positive and form a partition of unity, the curve points are convex combinations of the control points, and the curve is contained in the convex hull of the control polygon.

\begin{figure}[h]
	\centering
	\begin{subfigure}{0.47\textwidth}
		\includegraphics[width=\textwidth, keepaspectratio]{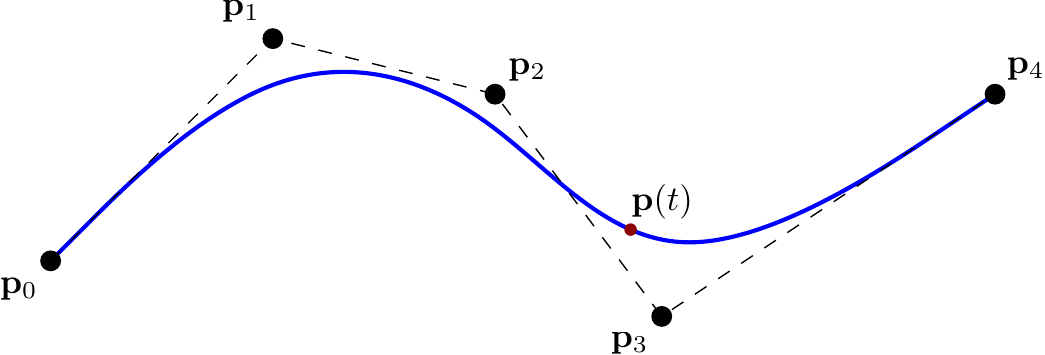}
		\subcaption{B\'{e}zier curve of degree 4}
	\end{subfigure}\ \ 
	\begin{subfigure}{0.3\textwidth}
		\includegraphics[width=\textwidth, keepaspectratio]{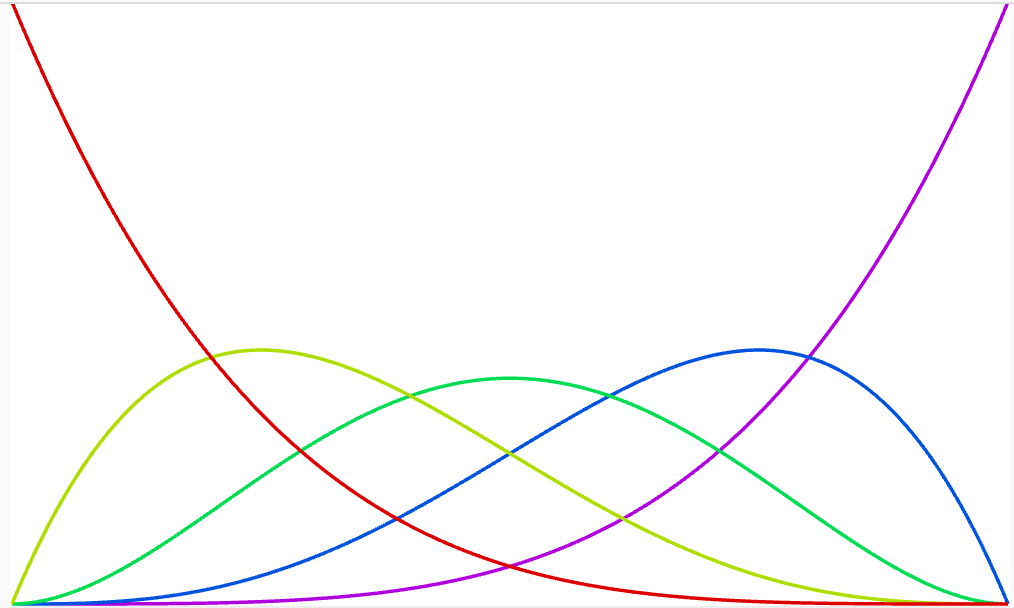}
		\subcaption{Bernstein basis of degree 4.}
	\end{subfigure}
	\caption{B\'{e}zier curves and Bernstein bases}
	\label{fig:beziercurve}
\end{figure}

\subsection{Polar Forms}\label{sec:motivation:polar}
The modern theory of B\'{e}zier curves and Bernstein bases is based on the concept of polar forms \cite{ ramshaw1987blossoming,ramshaw1989blossoms,goldman2003polar}. Given a univariate polynomial $F(t)$ of degree $d$, one can associate to $F(t)$ a $d$-variate function $f(t_{1},\ldots t_{d})$, called its \emph{polar form} (or \emph{blossom}) uniquely defined by the following properties:
\begin{itemize}
	\item \emph{multiaffinity} -- $f(\ldots,\lambda t_{i}+\mu s_{i},\ldots) = \lambda f(\ldots,t_{i},\ldots) + \mu f(\ldots,s_{i},\ldots),\quad  1 \leq i \leq d,\quad \lambda + \mu = 1$ 
	\item \emph{symmetry} -- $f(t_{1},t_{2},\ldots,t_{d}) = f(t_{\sigma(1)},t_{\sigma(2)},\ldots,t_{\sigma(d)})$, where $\sigma$ is any permutation on $d$ elements.
	\item \emph{diagonal property} -- $f(t,t,\ldots,t) = F(t)$
\end{itemize}
%Such a function can be produced via \emph{polarization}: replacing each monomial term $t^{k}$ with the average of all possible $k$-term products of $d$ different variables.

For a degree-$d$ B\'{e}zier curve $\mathbf{P}(t) = \sum_{i=0}^{d}\beta^{d}_{i}(t)\mathbf{P}_{i}$, the vector of polar forms $\mathbf{p}(t_{1},t_{2}\ldots,t_{d})$ evaluates to the curve itself on the diagonal, and gives the control points for special inputs:
\begin{align}\label{eq:polarcontrol}
\mathbf{p}(\underbrace{0, \ldots, 0}_{d-i}, \underbrace{1,\ldots,1}_{i}) = \mathbf{P}_{i}
\end{align}
This property allows us to assign \emph{polar labels} to the control points -- see \autoref{fig:polar}.  Combined with the defining properties, efficient algorithms can be devised for evaluation, differentiation, subdivision and degree elevation.

\begin{figure}[h]
	\centering
	\includegraphics[width=0.6\textwidth, keepaspectratio]{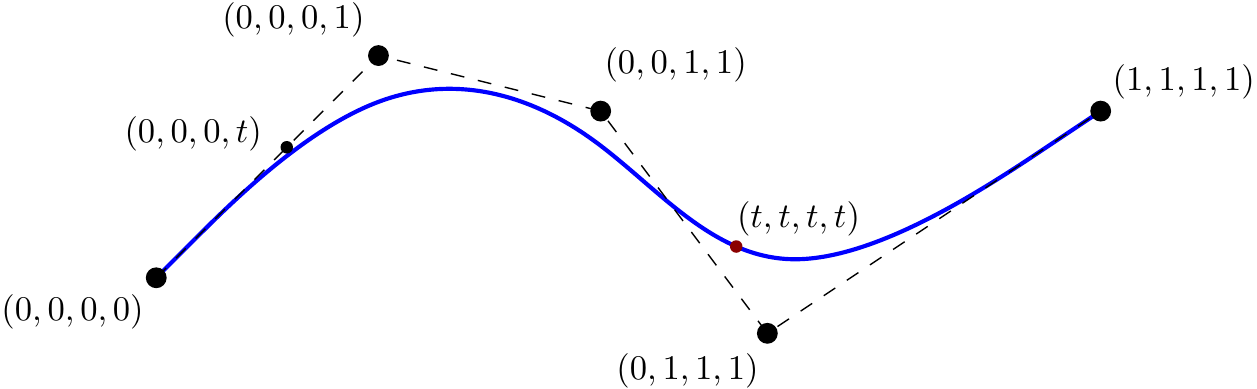}
	\caption{Degree-4 B\'{e}zier curve and control polygon, with polar labels}
	\label{fig:polar}
\end{figure}

Polar forms arise from the mathematical theory of \emph{symmetric tensor products} -- these have been developed into a general  framework by Ramshaw \cite{ramshaw2001paired}.

\subsection{Algebraic Geometry -- Normal Curves and Toric Varieties}\label{sec:motivation:newton}

Bernstein-B\'{e}zier representations are also related to algebraic geometry, in particular to toric varieties and Newton polytopes \cite{sottile2015algebraic}. Let us consider Bernstein polynomials in a slightly generalized way, by extending their domain from $\mathbb{R}$ to the entire projective line $\mathbb{RP}^{1}$.  This extension means reinterpreting the terms $t,1-t$ as homogeneous barycentric coordinates $[u:v] = [t:1-t]$. Bernstein polynomials are then proportional to the  degree-$d$ monomials in two variables: $\beta^{d}_{i}(u,v) \propto u^{i}v^{d-i}$. Define a parametric curve in $d$-dimensional projective space, with homogeneous coordinates parameterized by the degree-$d$ Bernstein polynomials:
\begin{align}\label{eq:veronese} \notag \mathbb{RP}^{1} & \longrightarrow \mathbb{RP}^{d} \\
\left[u : v\right] & \longmapsto \underbrace{[u^{d} : {d \choose 1}u^{d-1}v : {d \choose 2}u^{d-2}v^{2} : \ldots : {d \choose d-1}u v^{d-1} : v^{d}]}_{\text{all degree-$d$ monomials in $u,v$}}. \end{align}
This curve is called the (rational) \emph{normal curve} of degree $d$. These curves are prototypes of B\'{e}zier curves -- any B\'{e}zier curve of a given degree in any dimensions is an affine projection of the corresponding normal curve, see \autoref{fig:normalcurve}. Normal curves have been employed in the past for the purposes of curve classification \cite{degen1988some,pottmann1992classification}, to give a geometric interpretation of polar forms \cite{mazure1999blossoming}, to construct splines with geometric continuity \cite{seidel1993polar}, and to generalize B\'{e}zier curves to fractals \cite{goldman2004fractal}, or function spaces other than polynomials \cite{pottmann1993geometry}. 

\begin{figure}[h]
	\centering
	\includegraphics[width=0.6\textwidth, keepaspectratio]{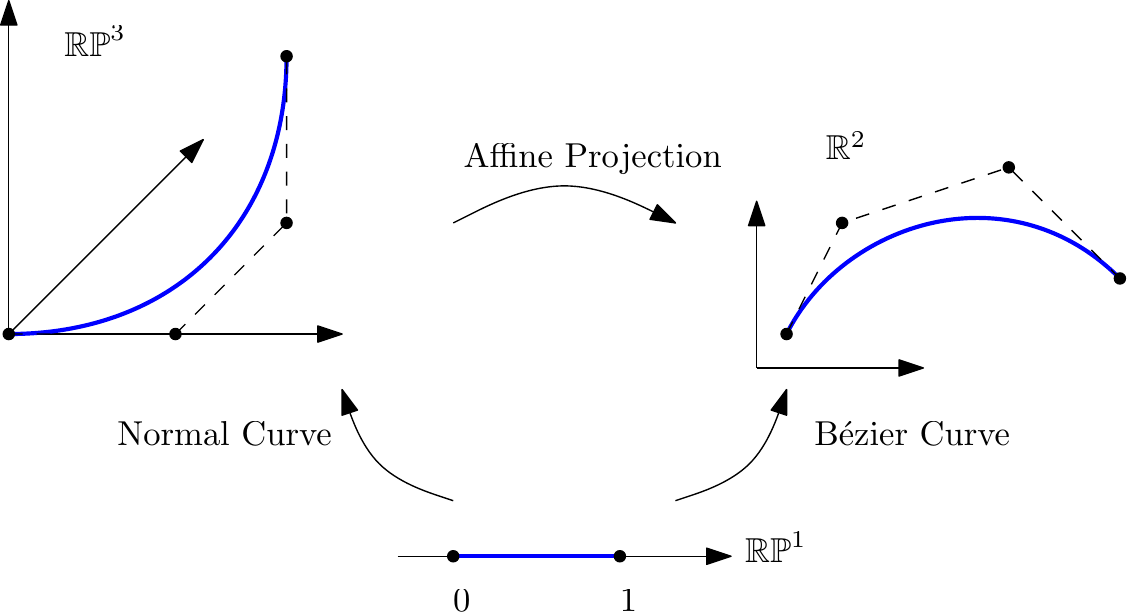}
	\caption{Normal curves}
	\label{fig:normalcurve}
\end{figure}

Algebraic varieties parameterized by monomial functions, such as rational normal curves, are examples of \emph{toric varieties}, which are usually studied over the complex numbers \cite{danilov1978geometry,cox2003what,cox2011toric}. The positive real slice of a complex toric variety is homeomorphic to its original parametric domain via the \emph{algebraic moment map} \cite{sottile2003toric}. The same is true for B\'{e}zier tensor products, simplices, and toric varieties in general. Toric varieties constructed from general subsets of monomials are used to define multi-sided generalizations of B\'{e}zier surfaces by Zub\'{e} and Krasauskas \cite{zub2000n,krasauskas2001shape,krasauskas2002toric} who introduced \emph{toric patches}, generalizing earlier ideas of Warren \cite{warren1992creating,warren1994multi,warren1994bound}. Toric variety theory has been employed in many works to analyze the properties of Bernstein-B\'{e}zier representations and their generalizations \cite{karciauskas1999comparison,cox2003universal,krasauskas2005universal,krasauskas2006bezier,craciun2010some,garcia2010linear,sottile2011injectivity,garc2011toric,sun2015g1}.

The relation with toric varieties also connects Bernstein-Bezier representations to Newton polytopes \cite{khovanskii1992newton,atiyah1983angular}. The Newton polytope of a polynomial in $n$ variables is the convex hull of the integer lattice points in $\mathbb{Z}^{n} \subset \mathbb{R}^{n}$ defined by the exponents of monomials with non-zero coefficients.\footnote{Note that for homogeneous polynomials, the Newton polytope is only $(n-1)$-dimensional.} For a B\'{e}zier curve of degree $d$ the Newton polytope is the segment $[0,d]$, which is also its natural parametric domain  -- see \autoref{fig:newtonpoly}.  The interplay between algebraic geometry and polytopes is characteristic of toric varieties \cite{ewald1996combinatorial}.

\begin{figure}[h]
	\centering
	\begin{subfigure}{0.47\textwidth}
		\includegraphics[width=\textwidth, keepaspectratio]{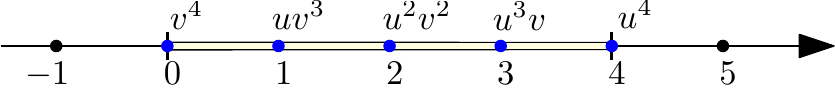}
		\subcaption{Newton polytope of a degree-4 B\'{e}zier curve}
	\end{subfigure}

	\begin{subfigure}{0.28\textwidth}
		\includegraphics[width=\textwidth, keepaspectratio]{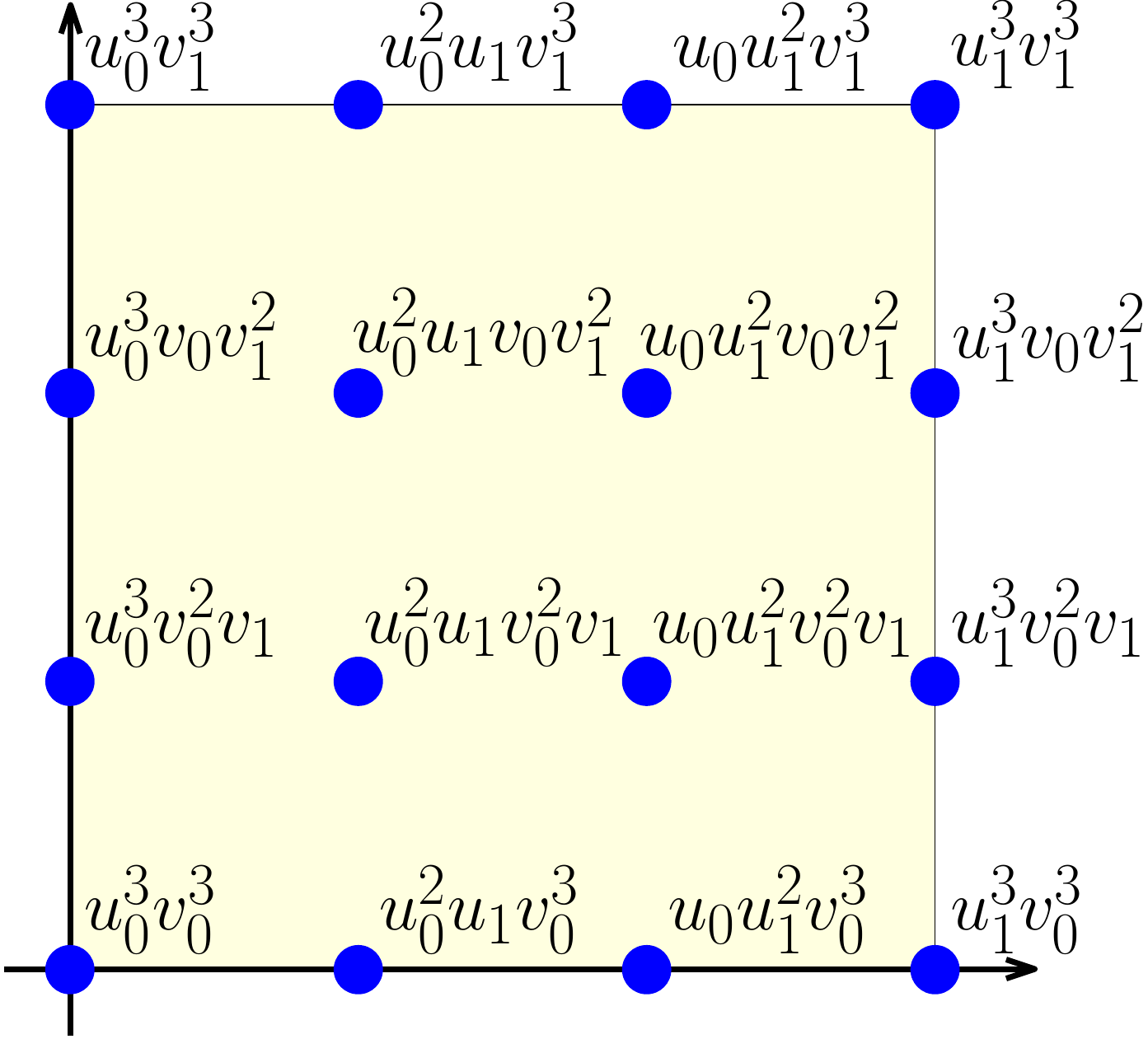}
		\subcaption{Newton polytope of B\'{e}zier tensor product of bi-degree $(3,3)$}
	\end{subfigure}\ \ \ \ 
	\begin{subfigure}{0.28\textwidth}
		\includegraphics[width=\textwidth, keepaspectratio]{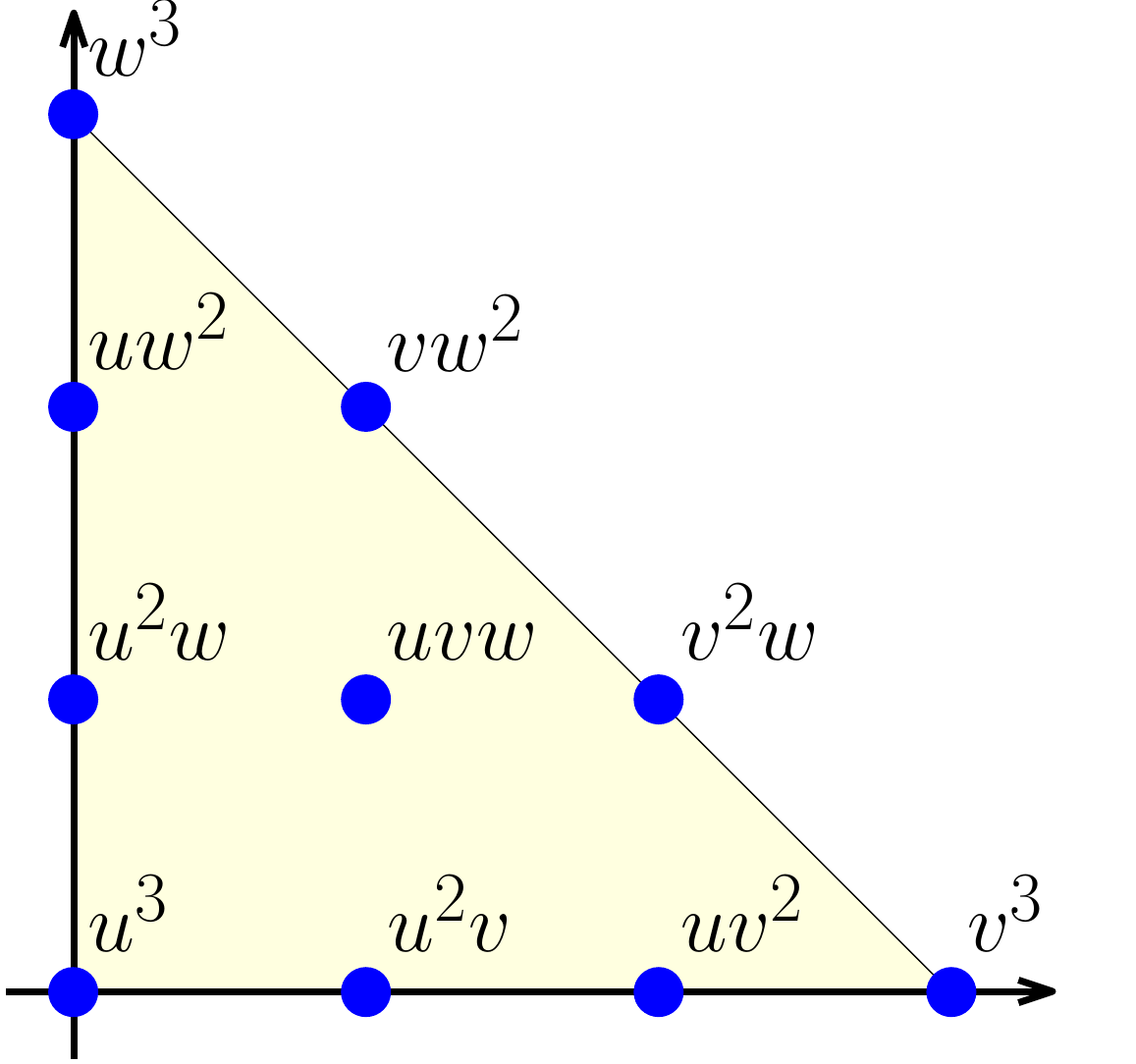}
		\subcaption{Newton polytope of B\'{e}zier triangle of degree $3$}
	\end{subfigure}
	\caption{Newton polytopes}
	\label{fig:newtonpoly}
\end{figure}

\subsection{Probability}\label{sec:motivation:prob}
B\'{e}zier curves have an obvious probabilistic interpretation in terms of \emph{binomial distributions}. Given a binary random variable, such as the outcome of a (possibly biased) coin flip, with probabilities $P(\mathrm{heads}) = p$ and $P(\mathrm{tails}) = 1-p$, the probability of getting exactly $k$ heads, out of $n$ trials is $P(\left| \mathrm{heads}\right| = k) = {n \choose k} (1-p)^{n-k}p^{k} $, which is a Bernstein polynomial. The control points of a degree-$d$ curve thus represent the possible outcomes (values of the random variable) after $n = d$ trials, and a point on the curve for parameter $t$ is a convex combination (expectation) corresponding to a binomial distribution with probability $p = t$. Goldman considered various generalizations of B\'{e}zier curves using P\'{o}lya's urn models \cite{goldman1985polya}. Related generalizations were proposed using umbral calculus \cite{winkel2001generalized,winkel2014generalization,winkel2015generalization,winkel2016generalization} and quantum\footnote{"Quantum" in this context appears to be unrelated to the quantum mechanics relevant to our work.} calculus \cite{orucc2003q,goldman2012formulas,goldman2015quantum}.

\begin{figure}[h]
	\centering
	\includegraphics[width=0.35\textwidth, keepaspectratio]{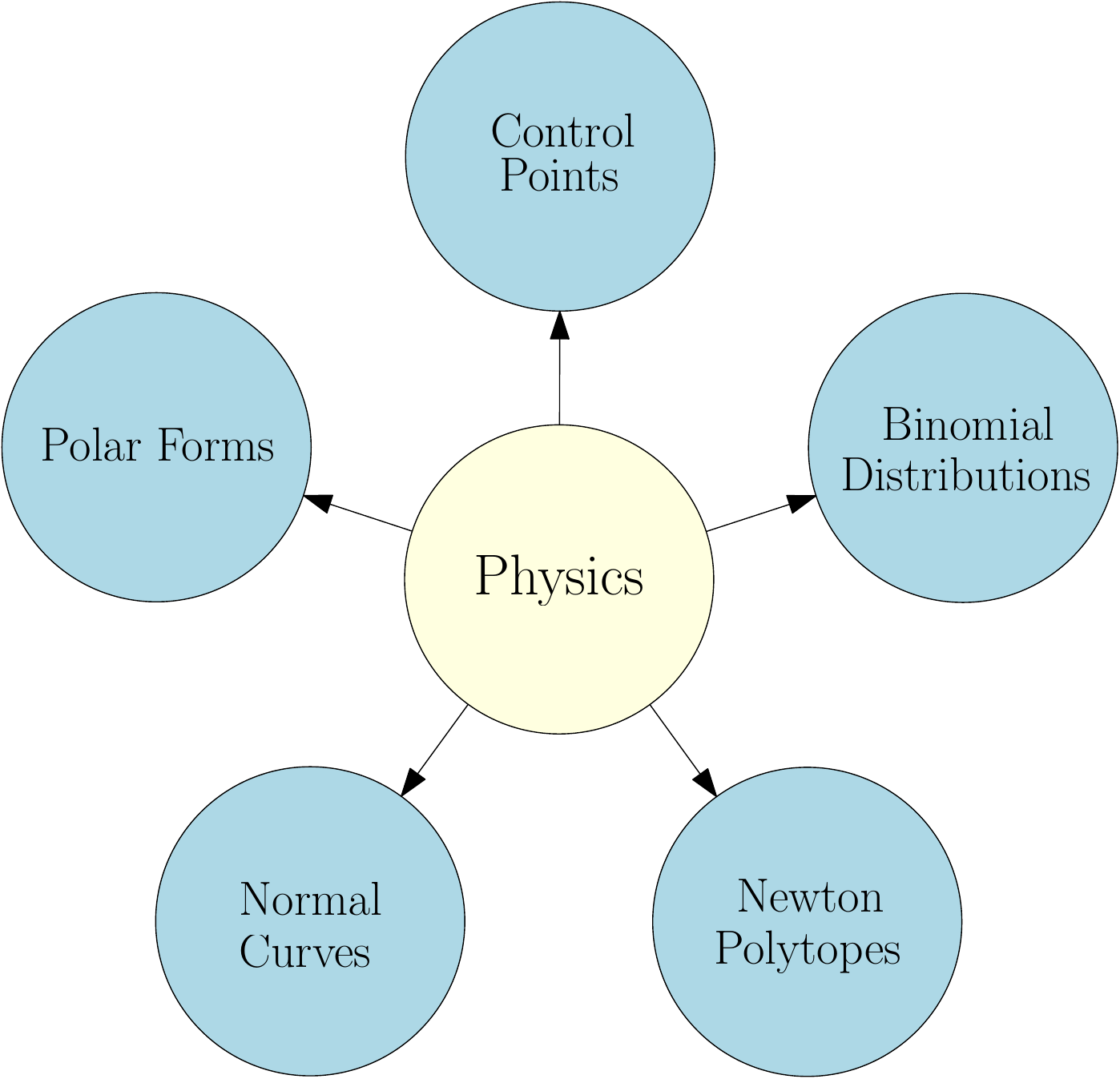}
	\caption{Aspects of B\'{e}zier theory arise from physics}
	\label{fig:circles}
\end{figure}

\section{The Physics of Rotational Motion}\label{sec:rotational}
We propose a theoretical framework that incorporates all the different aspects of Bernstein-B\'{e}zier theory. This framework is based on concepts from theoretical physics, with both classical and quantum mechanics playing a role -- see \autoref{fig:circles}. The relationship between toric varieties and physics has been explored in depth by physicists and mathematicians \cite{atiyah1983angular,guillemin1994moment,da2003symplectic}, but it is yet to be utilized in the context of CAGD. Geometric constructions involving normal curves have been used to analyse quantum mechanical systems \cite{brody2001geometric,brody2012quantum} -- see also the monograph \cite{bengtsson2006geometry} -- the relations with the Bernstein-B\'{e}zier theory, however, are not discussed. 

\subsection{Classical Angular Momentum  and Precession}\label{sec:rotational:classical}
The \emph{angular momentum} with respect to the origin of a particle with mass $m$ and velocity vector $\mathbf{v}$, is defined as the cross product of its position vector $\mathbf{r}$ and momentum vector $\mathbf{p} = m\mathbf{v}$:
\begin{align}
\mathbf{L} = \begin{bmatrix}
L_{x} \\
L_{y} \\
L_{z}
\end{bmatrix} := \mathbf{r} \times \mathbf{p}
\end{align}
Angular momentum is a 3-dimensional vector pointing in the direction of the axis of rotation (with counter-clockwise motion defined as positive), and its magnitude is proportional to the speed of rotation and the moment of inertia. For rigid bodies,  the constitutive angular momenta are integrated into a single vector -- see \autoref{fig:gyro:def}.

\begin{figure}[h]
	
	\centering
	\begin{subfigure}{0.25\textwidth}
		\includegraphics[width=\textwidth, keepaspectratio]{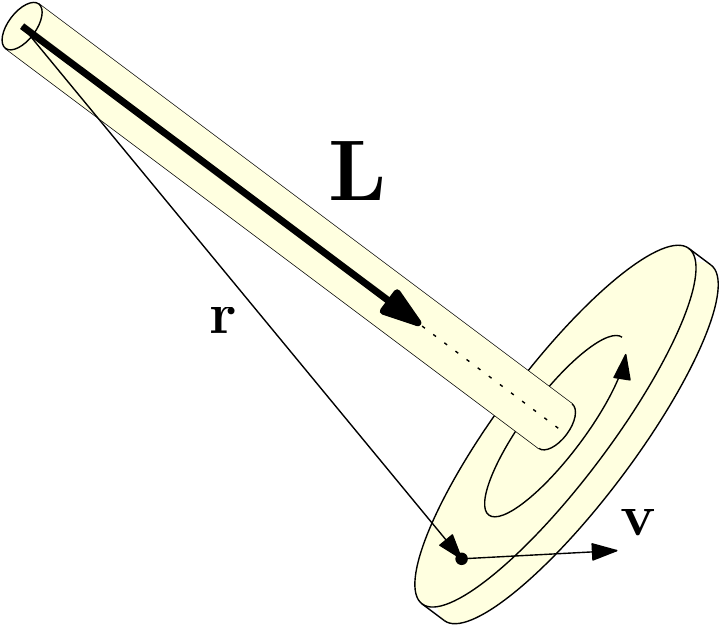}
		\subcaption{Definition}\label{fig:gyro:def}
	\end{subfigure}\ \ \ \ \ \ 
	\begin{subfigure}{0.32\textwidth}
		\includegraphics[width=\textwidth, keepaspectratio]{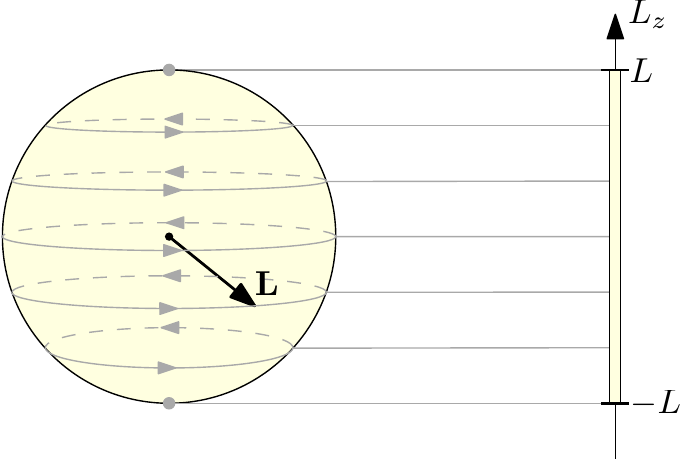}
		\subcaption{Precession}\label{fig:gyro:prec}
	\end{subfigure}
	\caption{Angular momentum}
	\label{fig:gyro}
\end{figure}

A well-known phenomenon involving angular momentum is the \emph{precession} of a spinning top or a gyroscope, i.e. when a rigid body spinning with some angular momentum about an axis is  also subjected to torque due to e.g. gravity. The torque acts perpendicular to the angular momentum and causes the spin axis to precess about the direction of the applied force.

Assuming the angular momentum magnitude to be a fixed value $L = \left| \mathbf{L} \right| $, the precession can be visualized as circular motion along a latitude of a sphere of radius $L$. The precession latitude depends on the $z$-component of the angular momentum vector, which is a height function over the sphere -- see \autoref{fig:gyro:prec}. We will later demonstrate that the range of possible values is related to the domain of a B\'{e}zier curve.

\subsection{Spin and Basics of Quantum Mechanics}\label{sec:rotational:quantum}
We next describe how angular momentum is described by quantum mechanics -- see e.g. \cite{penrose2006road,townsend2000modern}. Elementary particles, such as electrons have an intrinsic angular momentum, called \emph{spin} \cite{weiss2001spin}. A spinning particle that is electrically charged will also have a magnetic moment in the direction of the rotation axis. Such a particle behaves as a bar magnet when put in a magnetic field, i.e. the particle experiences torque aligning its moment with the applied field.

Classical physics suggests that when a spinning particle is moved through the poles of a magnet with a certain non-homogeneous magnetic field, the particle gets deflected from its trajectory, depending on how well its magnetic moment is aligned with the external force lines. At microscopic scales, however, physics is quite different. When a spinning particle moves through such a magnetic field, it deflects vertically either upwards or downwards by a fixed amount, as if the component of its angular momentum in the $z$-direction would be either $L_z = +\frac{1}{2}$ ("spin-up") or $L_z = -\frac{1}{2}$ ("spin-down"), with no other possibilities\footnote{We adopt natural units, so that Planck's constant is $\hbar = 1$.} -- see \autoref{fig:spin12}. The two values are measured with certain probabilities $p^{z}_{\uparrow}, p^{z}_{\downarrow}$, which depend on the initial orientation of the spin axis. Furthermore, after the measurement, the spin axis changes according to the (random) measured value.

\begin{figure}[h]
	\begin{subfigure}{0.48\textwidth}
		\centering
		\includegraphics[width=\textwidth, keepaspectratio]{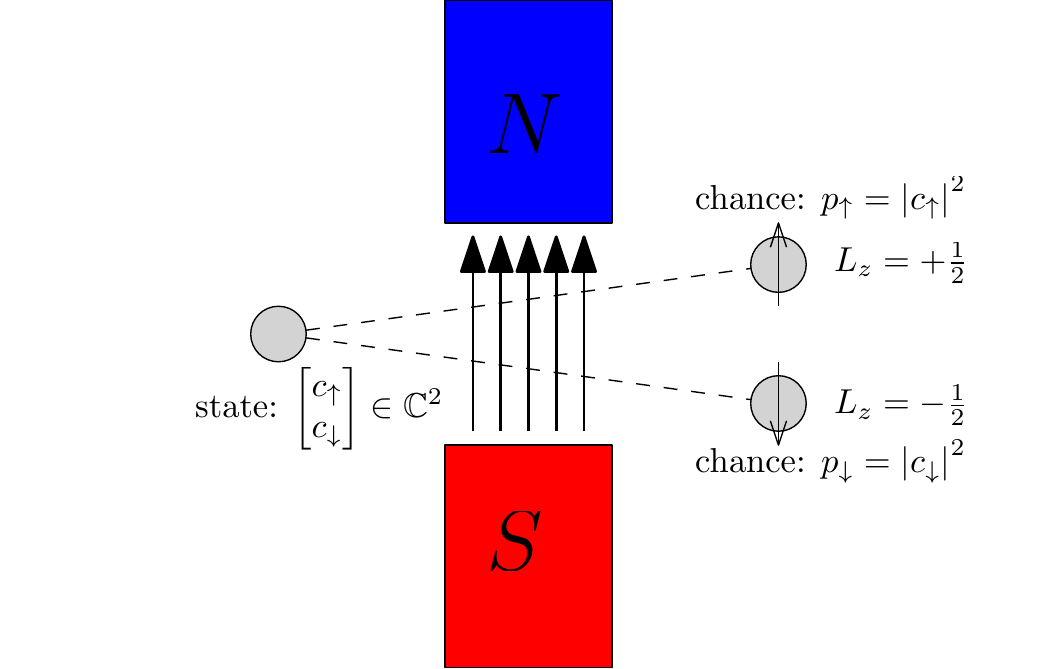}
		\subcaption{Spin-$\frac{1}{2}$ system}
		\label{fig:spin12}
	\end{subfigure}
	\begin{subfigure}{0.48\textwidth}
		\centering
		\includegraphics[width=\textwidth, keepaspectratio]{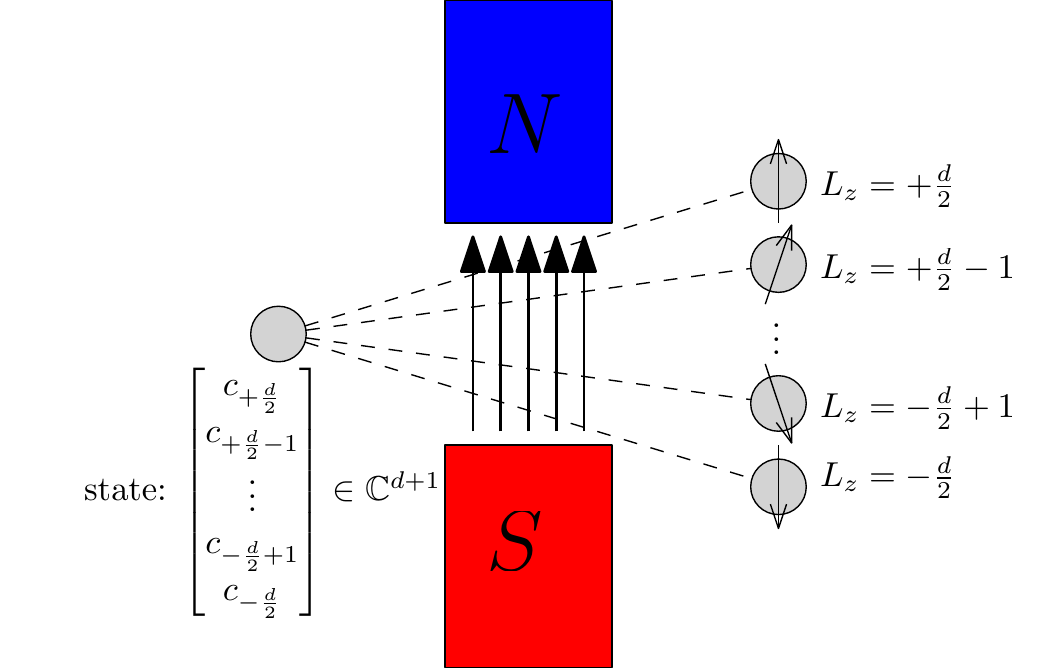}
		\subcaption{Spin-$\frac{d}{2}$ system}
		\label{fig:spind2}
	\end{subfigure}
	\caption{Measuring a component of spin}
\end{figure}

The same phenomena are observed along any other measurement axis, and the probabilities for the different axes are not independent. In fact, the probability distribution along \emph{any and all} axes can be encoded using only two complex numbers $c_{\uparrow}, c_{\downarrow}$, called \emph{probability amplitudes}, the squared magnitudes of which equal the probabilities: $p_{\uparrow} = \left| c_{\uparrow} \right|^{2}, p_{\downarrow} = \left| c_{\downarrow} \right|^{2}$. Thus, the quantum mechanical state of the electron spin is fully specified by a 2-dimensional complex vector
\begin{align}
\ket{\psi} = \begin{bmatrix}
c_{\uparrow} \\
c_{\downarrow}
\end{bmatrix} \in \mathbb{C}^{2}.
\end{align}
We use the standard \emph{bra-ket} notation for quantum state vectors -- a "ket" $\ket{\psi}$ denotes a column vector, a corresponding "bra" is the conjugate transpose $\bra{\psi} = (\overline{\ket{\psi}})^{T}$, and a "bra-ket" $\braket{\psi|\varphi}$ is an inner product.

Formally, a quantum mechanical state is specified by an element of a \emph{Hilbert space} -- a (possibly infinite-dimensional) complex vector space $\mathcal{H}$, equipped with an inner product. A priori, the elements of these Hilbert spaces are just abstract vectors, with no canonical coordinate representation. Measuring some observable, like the angular momentum component along an axis, implies choosing an orthonormal basis for the Hilbert space and expressing its elements in terms of complex coordinates. These complex coordinates are probability amplitudes, encoding the probabilities of measuring each of the possible values for the observable. The basis vectors $ \left\lbrace \ket{\psi^{i}} \in \mathcal{H}\right\rbrace  $ together with the corresponding measured values $ \left\lbrace \lambda_{i} \in \mathbb{R} \right\rbrace  $ are interpreted as the \emph{eigenvectors} and \emph{eigenvalues} of a \emph{linear operator} (matrix) $O: \mathcal{H} \rightarrow \mathcal{H}$, which are self-adjoint (or Hermitian) since the eigenvalues are real numbers, and the eigenvectors form an orthonormal set:
\begin{align}
O \ket{\psi^{i}} &= \lambda_{i}\ket{\psi^{i}}, \\
(\bar{O})^{T} &= O.
\end{align}
The basis states are thus referred to as \emph{eigenstates}. %The quadratic form associated to an operator evaluated for a state vector $\ket{\psi}$ gives the expectation value (or average) of the corresponding observable for that state:
%\begin{align}
%\bra{\psi}O\ket{\psi} = \left(\sum_{i = 1}^{d}{\bar{c_{i}}\bra{\psi^{i}}}\right)O\left(\sum_{i = 1}^{d}{c_{i}\ket{\psi^{i}}}\right) = \sum_{i = 1}^{d}{\left| c_{i} \right|^{2}\lambda_{i} = \mathbb{E}(\lambda) }.
%\end{align}

The electron spin is an example of a spin-{$\frac{1}{2}$} quantum system, also known as a \emph{qubit} (quantum bit). These qubits are described by a 2-dimensional Hilbert space, and any observable -- corresponding to the spin components along some direction -- can be expressed in terms of the \emph{Pauli spin matrices} representing measurements along the $x,y,z$ axes:
\begin{align}\label{eq:paulispin}
\sigma_{x} 
&= 
\frac{1}{2}\begin{bmatrix}
0 & 1 \\ 
1 & 0 
\end{bmatrix} 
&
\sigma_{y}
&= 
\frac{1}{2}\begin{bmatrix}
0 & -i \\ 
i & 0 
\end{bmatrix} 
&
\sigma_{z}
&= 
\frac{1}{2}\begin{bmatrix}
1 & 0 \\ 
0 & -1 
\end{bmatrix}.
\end{align}
These matrices all have eigenvalues $\pm \frac{1}{2}$ and $\sigma_{z}$ in particular has eigenvectors 
\begin{align}
\ket{\psi_{z}^{+\frac{1}{2}}} 
&=  
\begin{bmatrix}
1 \\
0
\end{bmatrix} 
&
\ket{\psi_{z}^{-\frac{1}{2}}} 
&= 
\begin{bmatrix}
0 \\
1
\end{bmatrix}.
\end{align}
We adopt the standard notation
\begin{align}
\ket{\uparrow} &:= \ket{\psi_{z}^{+\frac{1}{2}}}  & \ket{\downarrow} &:= \ket{\psi_{z}^{-\frac{1}{2}}} 
\end{align}
for the \emph{spin-up} and \emph{spin-down} eigenstates along the $z$-axis. Thus, a spin-$\frac{1}{2}$ state can be written as
\begin{align}
\ket{\psi} = c_{\uparrow}\ket{\uparrow} + c_{\downarrow}\ket{\downarrow}.
\end{align}

Systems with spin higher than $\frac{1}{2}$ also exist -- all positive half-integers $\frac{d}{2}, d \in \mathbb{N}$ are possible values. When a component of angular momentum is measured for a spin-$\frac{d}{2}$ system, the result is one of the $d+1$ possible values $(-\frac{d}{2}, -\frac{d}{2}+1, \ldots, +\frac{d}{2}-1, +\frac{d}{2})$ -- see \autoref{fig:spind2}. The quantum mechanical state is described by an element of a $(d+1)$-dimensional Hilbert space, which can be expressed in terms of the eigenstates for e.g. the $z$-axis as
\begin{align}
\ket{\psi} = c_{-\frac{d}{2}}\ket{\psi_{z}^{-\frac{d}{2}}} + c_{-\frac{d}{2} + 1}\ket{\psi_{z}^{-\frac{d}{2} + 1}} + \ldots + c_{\frac{d}{2} -1}\ket{\psi_{z}^{\frac{d}{2}-1}} + c_{\frac{d}{2}}\ket{\psi_{z}^{\frac{d}{2}}}.
\end{align}

The letter $d$ has been chosen deliberately to suggest a relation with the degree of a B\'{e}zier curve. Indeed as we will show later, a degree-$d$ B\'{e}zier curve corresponds to a spin-$\frac{d}{2}$ quantum system, with the $d+1$ control points and Bernstein polynomials giving the $d+1$ probability amplitudes for some special subset of spin-$\frac{d}{2}$ quantum states.

\section{The Bloch Sphere and Coherent States}\label{sec:spinfacts}
\subsection{The Bloch Sphere Representation for Spin-$\frac{1}{2}$}\label{sec:spinfacts:bloch}
We present the standard method to visualize the quantum states of spin systems. A spin-$\frac{1}{2}$ system is described by two complex numbers, which encode probabilities, so their squared magnitudes sum to 1, and only the difference between their complex phases has a physical significance. This interpretation means that the quantum state can be written as
\begin{align}
\ket{\psi} &= 
\begin{bmatrix} z_{1} \\ z_{2} \end{bmatrix} \propto \begin{bmatrix} e^{i\varphi}\cos{\frac{\theta}{2}} \\ \sin{\frac{\theta}{2}} \end{bmatrix}.
\end{align}
The two angles $(\theta,\varphi)$ can be interpreted as spherical coordinates, parameterizing the surface of a unit sphere in 3D. Thus, the quantum state of a spin-$\frac{1}{2}$ can be associated to a point on the \emph{Bloch sphere} -- see \cite[Ch. 5.2]{bengtsson2006geometry}. The $\ket{\downarrow} = \begin{bmatrix}
0 \\ 
1
\end{bmatrix}$ 
and 
$\ket{\uparrow} = \begin{bmatrix}
1 \\ 
0
\end{bmatrix}$ eigenstates correspond to the South and North poles. See \autoref{fig:bloch:sphere}.

\begin{figure}[h]
	\centering
	\begin{subfigure}{0.32\textwidth}
		\includegraphics[width=\textwidth, keepaspectratio]{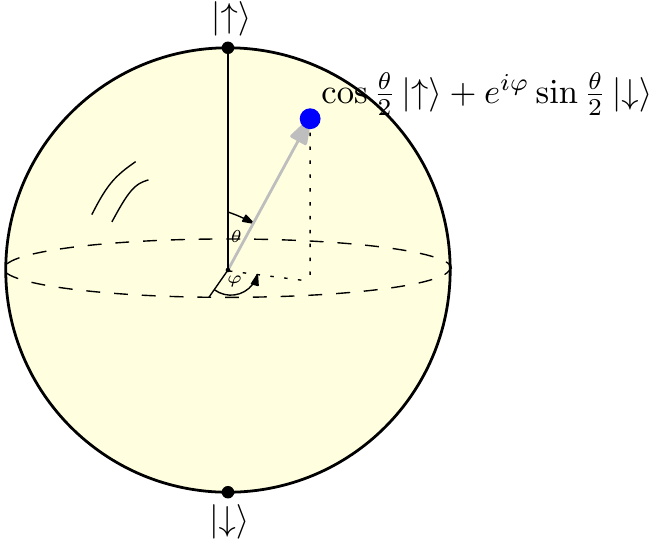}
		\subcaption{Bloch sphere}\label{fig:bloch:sphere}
	\end{subfigure}\ \ \ \ \ \ \ 
	\begin{subfigure}{0.38\textwidth}
		\includegraphics[width=\textwidth, keepaspectratio]{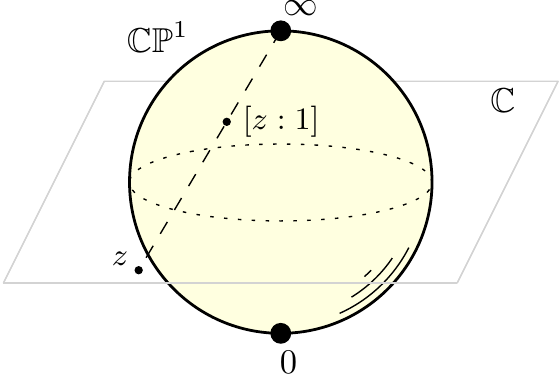}
		\subcaption{Stereographic projection onto a  sphere}\label{fig:bloch:stereo}
	\end{subfigure}

	\begin{subfigure}{0.6\textwidth}
		\includegraphics[width=\textwidth, keepaspectratio]{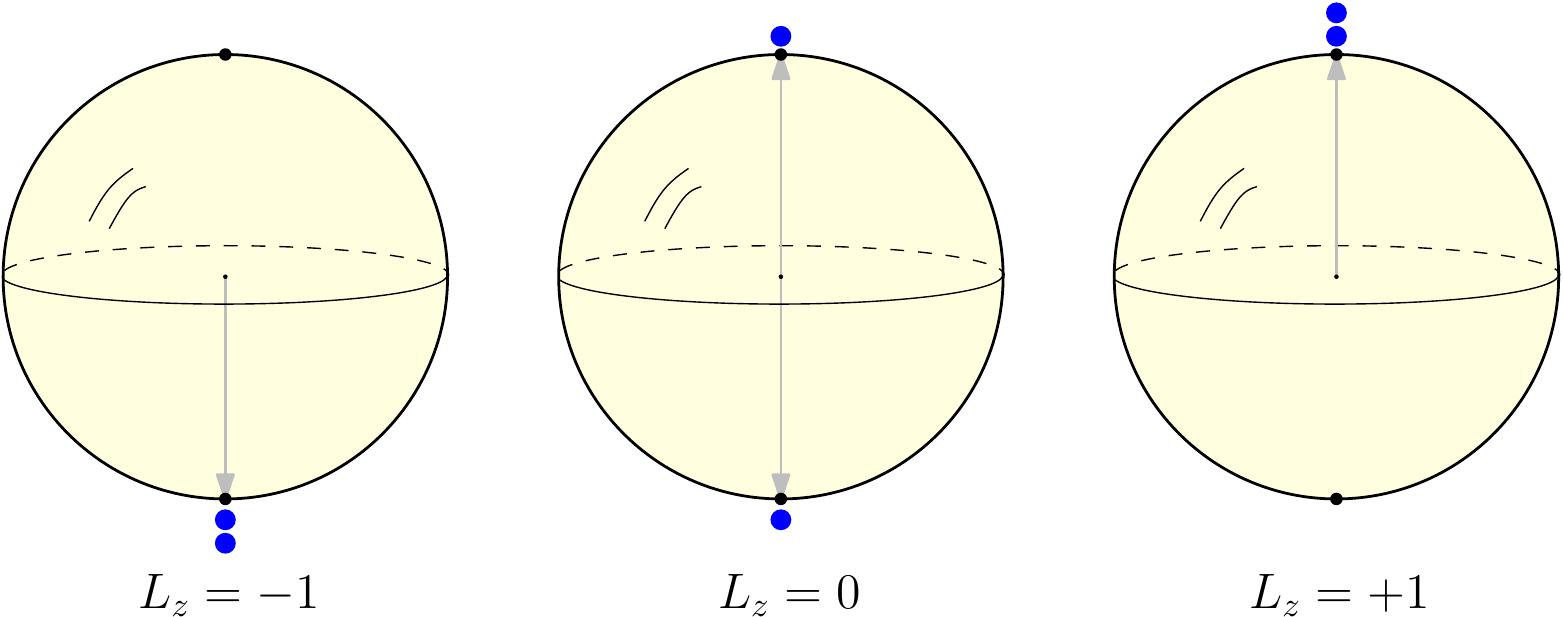}
		\subcaption{Spin-1 eigenstates on the Bloch sphere}\label{fig:bloch:spin1}
	\end{subfigure}
	\caption{}
	\label{fig:bloch}
\end{figure}

\subsection{The Bloch Sphere Representation for Spin-$\frac{d}{2}$}\label{sec:spinfacts:majorana}
A spin-$\frac{d}{2}$ system is described by a $(d+1)$ dimensional complex vector 
$\begin{bmatrix}
	c_{-{\frac{d}{2}}} & c_{-{\frac{d}{2}} + 1} & \ldots & c_{{\frac{d}{2}} + 1} & c_{{\frac{d}{2}}}
\end{bmatrix}$. 
Define a polynomial of degree $d$, with these $d+1$ numbers as coefficients:
\begin{align}\label{eq:spinpoly}
p(z) = c_{-{\frac{d}{2}}}z^{d} + c_{-{\frac{d}{2}}+1}z^{d-1} + \ldots +  c_{+{\frac{d}{2}}+1}z + c_{+\frac{d}{2}}
\end{align}
We now consider the stereographic projection of the complex plane (the complex 1-dimensional line) $\mathbb{C}$, onto the complex projective line $\mathbb{CP}^{1}$, the \emph{Riemann sphere}, parameterized by homogeneous coordinates $[z:w]$ -- see \autoref{fig:bloch:stereo}. The points of the complex plane $z \in \mathbb{C}$ map to points $\left[z : 1 \right]$, while the point at infinity maps to the North pole $[1 : 0]$. The polynomial \eqref{eq:spinpoly} is not  well-defined as a function on $\mathbb{CP}^{1}$, but we can still consider the zeroes of the equivalent homogeneous polynomial of $z$ and $w$:
\begin{align}\label{eq:spinhompoly}
c_{-{\frac{d}{2}}}z^{d} + c_{-{\frac{d}{2}}+1}z^{d-1}w + \ldots +  c_{+{\frac{d}{2}}+1}zw^{d-1} + c_{+\frac{d}{2}}w^{d} = 0.
\end{align}
As a consequence of the fundamental theorem of algebra, such a polynomial will have $d$ roots $(z_{1}, z_{2}, \ldots, z_{d})$ lying on the Riemann sphere (some potentially at infinity, with $w = 0$), -- these points (also called \emph{Majorana stars}) in turn uniquely define the polynomial coefficients and thus the quantum state (up to a scalar factor) \cite[Ch. 22.10]{penrose2006road}. Note that, being roots of a polynomial, the points have no natural ordering.

The eigenstates of a spin-$\frac{d}{2}$ system have only one non-zero polynomial coefficient, and thus all their roots are at the origin (south pole) or at infinity (north pole). Recalling the case of spin-$\frac{1}{2}$ systems and the Bloch sphere, we can label each of the spin-$\frac{d}{2}$ eigenstates by an unordered set of $d$ spin-$\frac{1}{2}$ eigenstates ($\ket{\uparrow}$ and $\ket{\downarrow}$).
\begin{align}
\ket{\psi_{z}^{-\frac{d}{2} + k}}  = \ket{\underbrace{\downarrow \ldots \downarrow}_{d-k}\underbrace{\uparrow \ldots \uparrow}_{k}}
\end{align}
The case of spin-1 eigenstates is shown in \autoref{fig:bloch:spin1}. Observe the similarity with the polar labels of B\'{e}zier control points \eqref{eq:polarcontrol}. The points being unordered is the analogue of the symmetry property of polar forms. We turn to the analogue of the diagonal property next.

\subsection{Spin Coherent States}\label{sec:spinfacts:coherent}
\begingroup
\setlength{\columnsep}{6pt}
\setlength{\intextsep}{0pt}
\begin{wrapfigure}{r}{0pt}
	\includegraphics[width=0.18\textwidth, keepaspectratio]{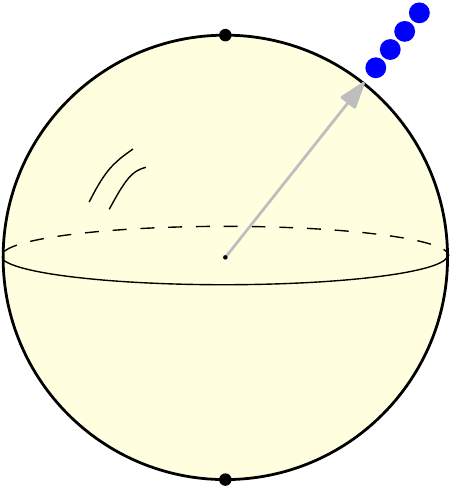}
	\label{fig:cylinder}
\end{wrapfigure}
A well-known feature of quantum mechanics is the Heisenberg uncertainty principle, which states that certain quantities -- the position and velocity of a particle, or its spin along independent axes --  cannot be determined simultaneously to arbitrary precision. This principle is expressed in terms of inequalities, imposing lower bounds on products of variances for different observables. States that attain these bounds are of much interest in physics, since their behaviour is as similar to classical mechanics as allowed by quantum mechanics. These states are called \emph{coherent states}. The formal definition of coherent states requires group theoretical and functional analytical concepts beyond the scope of our paper -- we refer the reader to \cite{gazeau2009coherent,combescure2012coherent} and \cite{bengtsson2006geometry}. The coherent states of a spin-$\frac{1}{2}$ system are easy to describe, since in this case every state happens to be coherent. This coherence is mirrored in the fact that these states are described by a single 3D vector, similar to classical angular momentum. Systems of higher spin can be described as collections of spin-$\frac{1}{2}$ systems, and states that have all $d$ spin-$\frac{1}{2}$ directions pointing in the same direction -- as shown in the inset for $d = 4$ -- are suspected to be the coherent states, in analogy with a polar form evaluating to the B\'{e}zier curve along the diagonal. This is in fact the case -- we refer to \cite[Ch. 7.2.]{bengtsson2006geometry} for a proof. Furthermore the following is true:

\endgroup
\begin{theorem}
	A coherent state of a spin-$\frac{d}{2}$ system defines a binomial probability distribution over the eigenstates with event probability $p = \cos^{2}\left( \frac{\theta}{2} \right)$ and $d$ trials.
\end{theorem} 
\begin{proof}
	A coherent state is represented by $d$ confluent points on the Riemann sphere, corresponding to a polynomial with a $d$-fold root
	\begin{align}
	p(z) = (z-r)^{d}.
	\end{align}
	This polynomial can be expanded using the binomial theorem as
	\begin{align}
	p(z) = \sum_{k = 0}^{d}\left({d \choose k}(-1)^{k}r^{d-k}\right)z^{k}.
	\end{align}
	Substituting the dehomogenized complex coordinate $r = \tan{\left( \frac{\theta}{2}\right) }e^{i\varphi}$ then multiplying by the (physically irrelevant) common factor $\cos^{d}\left( \frac{\theta}{2}\right) $, we get
	\begin{align} p(z) = \sum_{k = 0}^{d}\left({\sqrt{d \choose k}}(-1)^{k}\cos^{k}\left(\frac{\theta}{2}\right)\sin^{d-k}\left(\frac{\theta}{2}\right)e^{(d-k)i\varphi}\right)\sqrt{d \choose k}z^{k}. \end{align}
	The coefficients of the scaled monomials are probability amplitudes for the coherent state. The squared absolute values define a binomial probability distribution
	\begin{align} P(k,d) = \left| c_{k} \right|^{2} = {d \choose k}\left(\cos^{2}{\left(\frac{\theta}{2}\right)}\right)^{k}\left(\sin^{2}{\left(\frac{\theta}{2}\right)}\right)^{d-k},\end{align}
	with the claimed event probability.
\end{proof}

\section{B\'{e}zier Curves from Physics -- Summary}\label{sec:summary}
We have seen that the quantum mechanics of spin systems is analogous to the mathematics of B\'{e}zier curves. First, a spin-$\frac{d}{2}$ system and a degree-$d$ B\'{e}zier curve both have $d+1$ degrees of freedom and are related to probability distributions over $d+1$ possible outcomes (control points). A spin-$\frac{d}{2}$ system can be treated as a symmetric ensemble of spin-$\frac{1}{2}$ systems (qubits), which is analogous to how polar forms associate to a degree-$d$ B\'{e}zier curve a $d$-fold linear function. Polar forms evaluate to the control points for certain inputs, while spin eigenstates are equivalent to symmetric ensembles of spin-$\frac{1}{2}$ eigenstates (i.e. spin-ups and spin-downs). Thus, polar labels assigned to B\'{e}zier curves can be interpreted in terms of qubits, as shown in \autoref{fig:polarform_spin}. Spin coherent states are characterized by all the qubits coinciding and also by binomial probability distributions (Bernstein polynomials) -- analogous to the diagonal property of polar forms that defines points on the B\'{e}zier curve. The multiaffinity property of polar forms is reflected in the linearity of quantum state spaces, i.e. that quantum systems can be in linear superpositions of different states. These correspondences are summarized in \autoref{tab:list}. Note that for spin-$\frac{1}{2}$ systems all states are coherent, which is compatible with the fact that a linear B\'{e}zier curve coincides with its control polygon, while for higher spin the only coherent eigenstates are indexed by $\pm\frac{d}{2}$, as a higher degree curve interpolates only its first and last control points. Furthermore, when a curve is degree-elevated iteratively, the control points converge to the curve itself -- on the physics side,  the quantum behaviour of spin-$\frac{d}{2}$ systems becomes dominated by the almost-classical coherent states as $d \rightarrow \infty$ (or equivalently $\hbar \rightarrow 0$) \cite[Ch. 7.5]{gazeau2009coherent}.

\begin{figure}[h]
	\centering
	\begin{subfigure}{0.6\textwidth}
		\includegraphics[width=\textwidth, keepaspectratio]{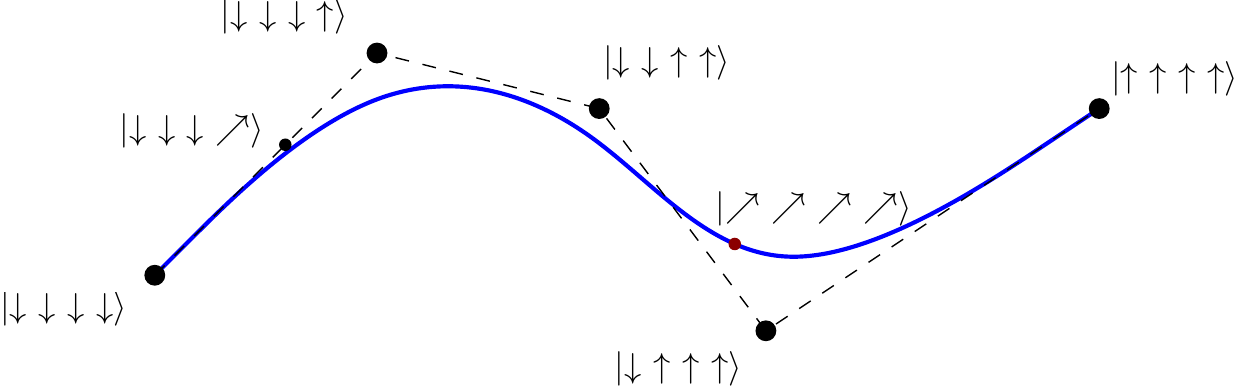}
	\end{subfigure}
	\caption{Polar labels in terms of spin-$\frac{1}{2}$ states}
	\label{fig:polarform_spin}
\end{figure}

\section{From B\'{e}zier Curves to Physics -- Overview}\label{sec:bezier2physics}
Previously, we demonstrated a correspondence between B\'{e}zier curves and theoretical physics. Based on these observations, we can ask the question: is there an underlying reason for this relationship? We will demonstrate that these analogies are consequences of rigorous mathematical arguments indicative of a deeper connection between geometric design and physics, so that a B\'{e}zier curve naturally contains a mechanical system. An overview of our argument can be seen on \autoref{fig:outline}. A more rigorous description of   our argument requires mathematical physical concepts beyond the scope of this paper and will be presented elsewhere \cite{vaitkus2018physics}. We resort to a high-level, informal overview and provide some pointers to the relevant mathematical literature.

First we turn from the actual B\'{e}zier curve in the plane or space, to the corresponding \emph{normal curve}, described in \autoref{sec:motivation:newton}, which is a parametric curve embedded in $d$-dimensional projective space. To work towards a physical interpretation, we extend the parametric domain from the real line to the \emph{complex plane}. The normal curve turns into a 2-dimensional surface, i.e. a sphere embedded in $d$-dimensional complex projective space \cite[Ch. 6.3]{bengtsson2006geometry}. This embedding allows the measurement of areas on the sphere, which is precisely the differential-geometric (\emph{symplectic}) structure that characterizes \emph{phase spaces} of classical mechanical systems in the Hamiltonian description \cite{arnold1989mathematical,hand1998analytical,guillemin1984symplectic}. The mechanical systems corresponding to complex normal curves are precessing \emph{gyroscopes} with a fixed angular momentum \cite[Sec. 3]{stone1989supersymmetry}. The Hamiltonian (total energy) is the natural height function on the sphere and the system evolves by rotations around the North-South axis. The range of energy values, called the \emph{(symplectic) moment map}, gives the parametric domain of the original curve \cite{sottile2003toric,da2003symplectic}. There exists a mathematical procedure called \emph{geometric quantization} that constructs the quantum mechanical equivalent of a classical system \cite{blau1992symplectic,todorov2012quantization,woodhouse1997geometric,hall2013quantum}. For B\'{e}zier curves of degree-$d$, quantization results in a spin-$\frac{d}{2}$ system \cite{nlab2016geometric} -- explaining the coincidences that we observed before. 

\begin{figure}[h]
	\centering
	\includegraphics[width=\textwidth, keepaspectratio]{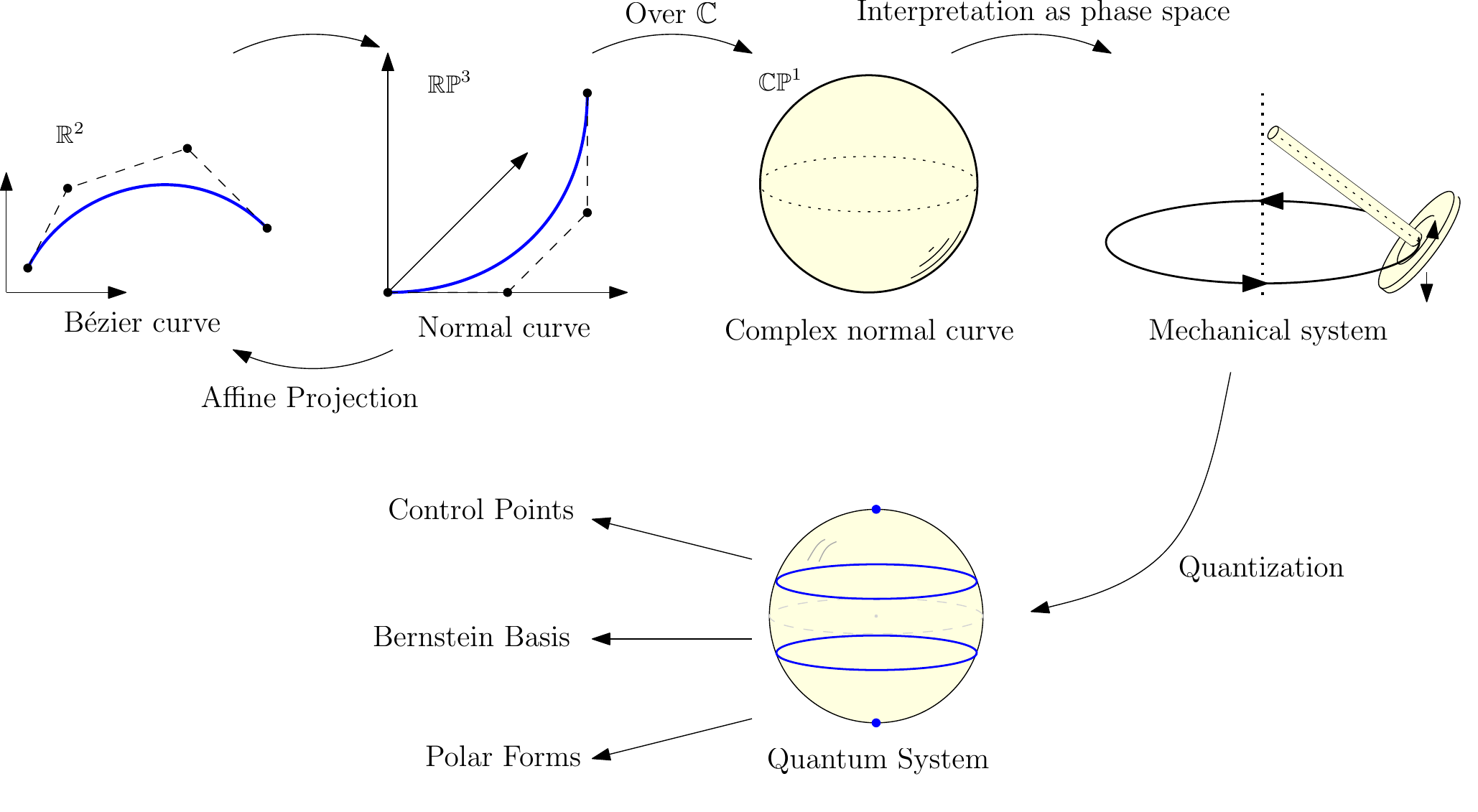}
	\caption{Overview of connection between B\'{e}zier curves and physical systems.}
	\label{fig:outline}
\end{figure}

\subsection{Complex Normal Curves}\label{sec:bezier2physics:Cnormal}
To reveal the physical systems contained in a B\'{e}zier curve, we  consider the corresponding normal curve, and extend its parametric domain to the complex plane. This extension turns the normal curve into a real 2-dimensional surface (a complex curve), which is topologically a 2-sphere via stereographic projection (\autoref{fig:bloch:stereo}).

\begingroup
\setlength{\columnsep}{6pt}
\setlength{\intextsep}{0pt}
\begin{wrapfigure}{r}{0.19\textwidth}
	\includegraphics[width=0.19\textwidth, keepaspectratio]{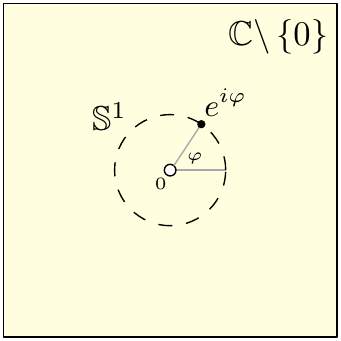}
	\label{fig:cylinder}
\end{wrapfigure}
The sphere contains the complex number plane, so multiplication by non-zero complex numbers $\mathbb{C}^{\ast} = \mathbb{C} \setminus \left\lbrace 0\right\rbrace$ moves points in a continuous way. As a consequence the numbers on the unit circle, i.e. complex numbers of unit magnitude $\mathbb{S}^{1} = \left\lbrace e^{i\varphi}, \varphi \in [0,2\pi) \right\rbrace \subset \mathbb{C}^{\ast}$ -- see the inset figure -- also act on the sphere. Their effect is a circular rotation, which can be visualized in a way that is identical to what we saw on \autoref{fig:gyro:prec} for a precessing gyroscope. This action suggests a connection with the physics of rotational motion, and is part of what we will explain in the subsequent chapters.

\endgroup
The complex normal curve is to be considered not just as some abstract manifold, as it is embedded into complex projective space using \eqref{eq:veronese}, which induces a metric geometric structure on it:
\begin{theorem}\label{thm:round}
	The degree-$d$ Bernstein polynomials embed $\mathbb{CP}^{1}$ into $\mathbb{CP}^d $ as a sphere of radius $\frac{d}{2}$.
\end{theorem}
\begin{proof}
	See \cite[Ch. 6.3]{bengtsson2006geometry}.
\end{proof}
What allows for a physical interpretation is that we can measure areas along the surface of the complex normal curve. Understanding why this is the case requires some basic notions from the mathematical theory of classical mechanics, which we describe in the next section.

\section{From Complex Normal Curves to Classical Mechanics}\label{sec:classicalbasics}
\subsection{Basics of Hamiltonian Mechanics}
We present only a high-level and non-technical overview of some concepts from theoretical mechanics. For further details, the reader is referred to the textbooks -- as an introduction we recommend \cite{susskind2013theoretical} and \cite{hand1998analytical}. Recall that Newton's 2nd law $\mathbf{F} = m\mathbf{a}$ prescribes accelerations -- the second time derivatives of coordinates. Given initial conditions (positions and first time derivatives), the state of a mechanical system can be computed by solving a second-order ODE. This observation implies that a mechanical system described by $d$ generalized coordinates (positions, angles, etc.) can be interpreted as a dynamical system on the $2d$-dimensional space of coordinates and their time derivatives. This viewpoint is taken in the Hamiltonian formulation of mechanics, where the mechanical system is represented by a \emph{phase space}, spanned by the coordinates $\mathbf{q} = (q_{1}, q_{2}, \ldots, q_{d})$ and corresponding (conjugate) momenta $\mathbf{p} = (p_{1}, p_{2}, \ldots, p_{d})$. As a consequence of \emph{conservation of energy}, there exists a scalar function $H(\mathbf{q},\mathbf{p})$ on the phase space, known as the total energy, or the \emph{Hamiltonian}, that remains constant in time:
\begin{align}
\frac{\partial H(\mathbf{q}(t),\mathbf{p}(t))}{\partial t} = 0
\end{align}

\begin{figure}[h]
	\centering
	\includegraphics[width=0.3\textwidth, keepaspectratio]{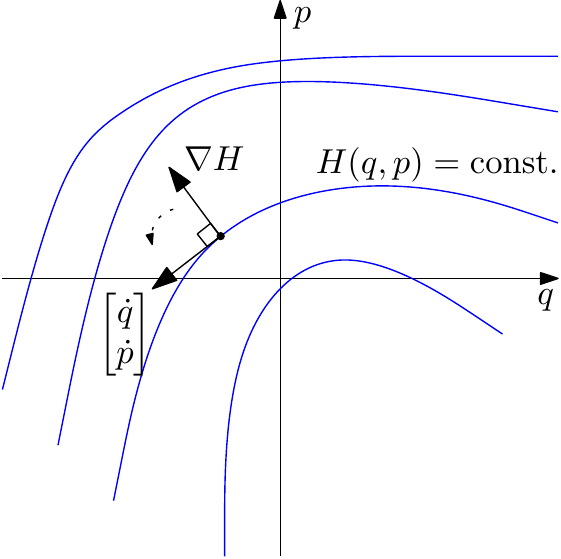}
	\caption{Hamiltonian mechanics in a two-dimensional phase space}
	\label{fig:hamiltonian}
\end{figure}

In other words, a classical mechanical system evolves in phase space along the level sets of its energy function. For a system described by a single generalized coordinate, i.e. a 2-dimensional phase space $\mathbb{R}^{2} = \left\lbrace(q,p) \right\rbrace$, with energy function $H(q,p)$, Newton's 2nd law implies that the time derivative is the gradient vector of the energy function, rotated by 90 degrees, as seen on \autoref{fig:hamiltonian}:
\begin{align}
\begin{bmatrix}
\dot{q} \\
\dot{p}
\end{bmatrix}
=
\begin{bmatrix}
0 & -1 \\
1 & 0
\end{bmatrix} \nabla H(q,p) = 
\begin{bmatrix}
-\frac{\partial H(q,p)}{\partial p} \\
\frac{\partial H(q,p)}{\partial q}
\end{bmatrix} 
\end{align}
These equations are \emph{Hamilton's equations} of motion.

The rotation matrix $\begin{bmatrix}
0 & -1 \\
1 & 0
\end{bmatrix}$ can also be interpreted as a bilinear form, in which case this matrix acts on pairs of phase space tangent vectors as follows:
\begin{align}
\omega(\mathbf{v},\mathbf{w}) = \begin{bmatrix}
q_{v} & p_{v}
\end{bmatrix}\begin{bmatrix}
0 & -1 \\
1 & 0
\end{bmatrix}\begin{bmatrix}
q_{w} \\
p_{w}
\end{bmatrix} = q_{w}p_{v} - q_{v}p_{w}
\end{align}
This expression is the signed area of the parallelogram spanned by the two vectors. This area form can be used to convert scalar functions (such as energy or angular momentum) on the phase space to vector fields (dynamical systems) tangent to the level sets. More importantly for what follows, we can go the other way and convert a vector field to a scalar function that stays constant along its integral curves.

Generalizing Hamiltonian mechanics from Euclidean phase spaces to surfaces such as the sphere, requires the machinery of \emph{symplectic geometry}, which studies manifolds with a symplectic structure (a generalization of the area form $\omega$). In this paper, we cannot describe this theory in detail  -- see e.g. \cite{mcduff2009what,arnold1990symplectic} for introductions.

\subsection{B\'{e}zier Curves as Classical Mechanical Systems}\label{sec:bezierclassical}
We have seen in \autoref{sec:bezier2physics:Cnormal} that a complex normal curve is a sphere, on which unit complex numbers induce a rotational motion. Infinitesimal rotations then define flows that preserve the area form. This observation suggests that the sphere is in fact a phase space of some mechanical system.
\begin{theorem}
	A complex normal curve of degree-$d$ is a Hamiltonian phase space, with the area form of a sphere with radius $\frac{d}{2}$. The energy function corresponding to rotations by unit complex numbers is proportional to the height function of the sphere.
\end{theorem}
\begin{proof}
	See \cite[Ch. 1.2]{da2003symplectic}
\end{proof}
The mechanical system thus associated to the degree-$d$ B\'{e}zier curve is a gyroscope with angular momentum $\mathbf{L}$ of fixed magnitude $\left| \mathbf{L} \right| = \frac{d}{2}$, precessing under  an external torque $\mathbf{T}$ \cite[Sec. 3]{stone1989supersymmetry}. The level sets of the energy $H = \mathbf{L} \cdot \mathbf{T}$ are the sphere latitudes. The assignment of phase space points to the values of quantities unchanged by a Hamiltonian flow is called the \emph{moment map} \cite{guillemin1994moment}. In our case the moment map is equivalent to the energy function, and its image is (some multiple of) the interval $[-\frac{d}{2}, \frac{d}{2}]$ \cite[Ch. 1.6]{da2003symplectic}. Up to translation, this is the domain and also the Newton polytope for a degree-$d$ B\'{e}zier curve. We now extend this connection between B\'ezier curves and classical mechanics further, to quantum mechanics.

\section{From Classical to Quantum Mechanics}\label{sec:bezierquantum}
Quantum mechanics is known to be radically different from classical mechanics. It is nevertheless common practice in theoretical physics to construct quantum systems from classical ones. This extension means taking a classical system, described by e.g. its phase space and assigning to it a vector space, along with operators representing the observable quantities. This process is known as \emph{quantization}\footnote{Not directly related to quantization in signal processing.}. We again resort to a high-level and informal overview.

\subsection{Geometric Quantization}
\begingroup
\setlength{\columnsep}{6pt}
\setlength{\intextsep}{0pt}
\begin{wrapfigure}{r}{0.48\textwidth}
	\includegraphics[width=0.48\textwidth, keepaspectratio]{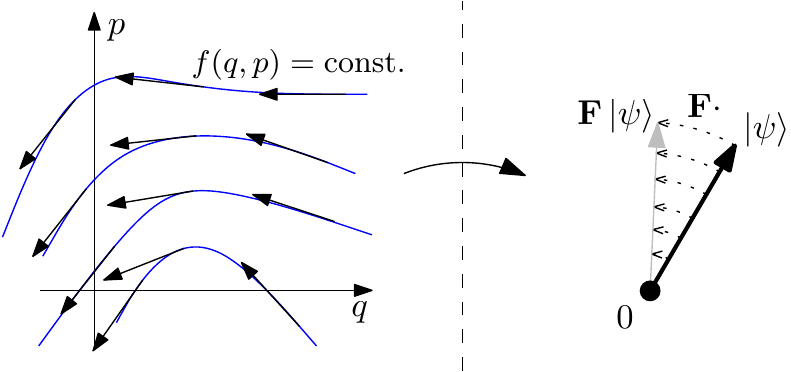}
	\label{fig:cylinder}
\end{wrapfigure}Let us assume we have a classical phase space, with generalized coordinates  $\mathbf{q} = (q_{1},\ldots,q_{n})$ and momenta $\mathbf{p} = (p_{1}, \ldots, p_{n})$. A classical observable is a scalar function $f(\mathbf{q},\mathbf{p})$, which corresponds to a flow along its level sets. Such a flow gives an infinitesimal transformation of phase space preserving the form of Hamilton's equations, so classical observables correspond to \emph{canonical transformations} \cite{hand1998analytical}. A change of coordinates must also have an effect on the corresponding quantum state. This transformation of the quantum state vector must preserve the sum of squared magnitudes of its components (i.e. probabilities) -- so it must be a unitary transformation. \emph{Infinitesimal} unitary matrices can be identified with \emph{Hermitian} matrices, which describe observables in quantum mechanics. In effect we have the correspondences
\begin{align*}
\text{classical observables} \longleftrightarrow \text{canonical transformations}
\longrightarrow \text{Hermitian matrices} \longleftrightarrow \text{quantum observables}
\end{align*} 

\endgroup

Observables also have a Lie-algebraic structure, which must be preserved by quantization \cite{nlab2016quantization}. This structure, along with other physical considerations, form a set of requirements every quantization procedure must adhere to \cite{todorov2012quantization}. While these requirements are impossible to satisfy in the general case (as expected, given the differences between classical and quantum physics), for many systems -- including gyroscopes -- there exist a well-defined procedure called \textit{geometric quantization}. See e.g. \cite[Ch. 22-23]{hall2013quantum} and \cite{blau1992symplectic} for details on this mathematically involved topic.

Geometric quantization rests on the observation that a tangent vector field on a manifold (e.g. a surface) represents a directional derivative, which is also a linear operator acting on the space of scalar functions. It is natural then to take as the quantum states  complex-valued functions over the phase space, on which the vector fields associated to scalar functions (classical observables) act by differentiation. Ensuring that this construction is both mathematically well-defined and stays compatible with the physics involves some advanced differential geometry and topology that space limitations prevent us from describing \footnote{In short, one needs a complex line bundle, with a curvature form compatible with the symplectic structure of the phase space. Certain special sections of this bundle, which form a finite-dimensional vector space, will serve as the quantum states.}. Generally, there are severe constraints on which phase spaces can be successfully quantized.  

\subsection{B\'{e}zier Curves as Quantum Mechanical Systems}
In the case most relevant for us, the phase space is a sphere, with an area form induced by its embedding into complex projective space. However, not all spherical phase spaces can be geometrically quantized.  
\begin{theorem}
	A sphere can be geometrically quantized as a phase space if and only if the sphere's radius is a half-integer $R = \frac{d}{2},\ d \in \mathbb{Z}.$
\end{theorem} 
\begin{proof}
	See \cite{nlab2016geometric}.
\end{proof}
Together with \autoref{thm:round}, this theorem implies that the complex normal curves associated to B\'{e}zier curves are exactly those spherical phase spaces that can be geometrically quantized.
\begin{theorem}
	The geometric quantization of a sphere with radius $\frac{d}{2}$ gives a $(d+1)$-dimensional\footnote{A minor complication arises in the metaplectic correction phase of quantization. This correction decreases the state space dimension, which can be compensated by increasing the sphere radius accordingly \cite{nlab2016geometric}.} Hilbert space, that describes a spin-$\frac{d}{2}$ quantum system. 
\end{theorem}
\begin{proof}
	See \cite[Ch. 4.2]{blau1992symplectic}.
\end{proof}
This theorem gives an explanation for the observed coincidences between B\'{e}zier curves and spin systems. 

\section{B\'{e}zier Curves and Harmonic Oscillators}\label{sec:bezier_osc}
Quantum spin systems can be described in terms of another quantum mechanical system: the quantum harmonic oscillator. This alternative description of B\'{e}zier theory will be key for further generalizations.
%\subsection{$SO(3)$, $SU(2)$ and $\mathfrak{su}(2)$}
%The physics of rotational motion is a consequence of the rotational symmetry of three-dimensional space. The group of 3D rotations $SO(3)$ has a well-known equivalent description in terms of unit quaternions, which can be identified with the group of $2 \times 2$ unitary matrices with unit determinant $SU(2)$. While the correspondence between rotations and quaternions is not one-to-one, in the context of quantum mechanics this ambiguity amounts to a physically irrelevant global phase factor, and thus $SU(2)$ is the relevant symmetry group for the quantum mechanics of rotational motion, i.e. spin. Indeed, the Pauli spin matrices represent the imaginary quaterions $i,j,k$. For physical applications, it is useful to work with the infinitesimal generators of the group forming a linear space with a special algebraic structure called a Lie algebra that encodes the non-commutativity of the underlying group via a bilinear, skew-commutative, non-associative operation called the Lie bracket . In case of rotations, the Lie algebra is simply the space of angular momentum vectors, and the Lie bracket is the usual 3D cross product. In terms of quaternions, we have a three-dimensional Lie algebra denoted as $\mathfrak{su}(2)$. As a consequence of bilinearity, a Lie algebra structure is unambigiously given by the brackets of its basis vectors -- for $\mathfrak{su}(2)$ we have.

\subsection{The Quantum Harmonic Oscillator}
The \emph{simple harmonic oscillator} models movement in a single direction about an equilibrium position, under a linear restoring force. An oscillator's state at any given time is described by the \emph{position} along a line $x \in \mathbb{R}$, and the linear \emph{momentum} $p \in \mathbb{R}$, and its total energy\footnote{The oscillator frequency is irrelevant for our purposes, so we have chosen the mass and the spring constant to be 1.} is $H := \frac{1}{2}(p^{2} + x^{2})$. Viewed in the $x-p$ plane, the oscillator's state stays on some circular level set $H(x,p) = E$, due to conservation of energy.

The quantum mechanical equivalent is the \emph{quantum harmonic oscillator} (QHO) \cite{townsend2000modern}. The quantum state of a QHO is described by a probability amplitude (wavefunction) over the possible positions. The observables $x$ and $p$ correspond to linear operators acting on functions over the real line:
\begin{align}
\mathbf{x} := x \cdot;\ \mathbf{p} := -i\frac{d}{dx}.
\end{align}
Operators generally do not commute -- position and momentum in particular satisfy the \emph{commutation relation}
\begin{align}
[\mathbf{x}, \mathbf{p}] := \mathbf{x}\mathbf{p} - \mathbf{p}\mathbf{x} = i.
\end{align}
A formal substitution of $\mathbf{x}$ and $\mathbf{p}$ into the expression of the energy results in the energy operator $\mathbf{H}$, which has a countably infinite set of eigenvectors (\emph{energy eigenstates}) $\mathbf{H}\ket{E_{n}} = E_{n}\ket{E_{n}}$ with eigenvalues $E_{n} =  n + \frac{1}{2},\ n = 0, 1, 2, \ldots.$

The energy of a QHO is thus quantized in the traditional sense and the energy eigenstates $\ket{n} := \ket{E_{n}}$ are characterized by the number of \emph{energy quanta} they contain. We define \emph{annihilation and creation operators}
\begin{align}\label{eq:creat_def}
\mathbf{a} := \frac{1}{\sqrt{2}}\left( \mathbf{x} + i \mathbf{p} \right);\ \mathbf{a}^{+} :=  \frac{1}{\sqrt{2}}\left( \mathbf{x} - i \mathbf{p} \right).
\end{align} 
These operators act on the energy eigenstates by removing or adding a single quantum of energy: 
\begin{align}\label{eq:creat_effect}
\mathbf{a}\ket{n} = \sqrt{n}\ket{n-1};\ \mathbf{a}^{+}\ket{n} = \sqrt{n+1}\ket{n+1}, 
\end{align}
and satisfy the commutation relations $[\mathbf{a},\mathbf{a}^{+}] = 1$.
We also define the \emph{number operator}, $\mathbf{N} := \mathbf{a}^{+}\mathbf{a}$, for which $\mathbf{N}\ket{n} = n\ket{n}$.
\subsection{Spin Systems as Pairs of Oscillators}
We now consider a pair of independent, identical harmonic oscillators. There are two sets of annihilation/creation operators $ \mathbf{a}_{i}, \mathbf{a}^{+}_{i}\  (i = 1,2)$ satisfying the relations $[\mathbf{a}_{i},\mathbf{a}_{j}] = [\mathbf{a}^{+}_{i}, \mathbf{a}^{+}_{j}] = [\mathbf{a}_{i}, \mathbf{a}^{+}_{j}] = 0\  (i \neq j).$
The following observation is due to J. Schwinger \cite{schwinger1952angular}:

\begin{theorem}
	A pair of identical harmonic oscillators containing exactly $d$ quanta of energy is quantum mechanically equivalent to a spin-$\frac{d}{2}$ system.
\end{theorem}
\begin{proof}
	Proving that two quantum mechanical systems are equivalent requires that their Hilbert spaces are of the same dimension, and there must exist a corresponding set of operators with the same eigenvalues and commutation relations.
	Given that $d$ quanta of energy can be distributed among two oscillators in $d+1$ different ways, the Hilbert space dimensions are equal. We define the operators $\mathbf{J}_{+} = \mathbf{a}_{1}^{+}\mathbf{a}_{2},\ \mathbf{J}_{-} = \mathbf{a}_{2}^{+}\mathbf{a}_{1}$, which redistribute a single quantum of energy between the oscillators, and $\mathbf{J}_{z} = \frac{1}{2}(\mathbf{N}_{1} - \mathbf{N}_{2}) =  \frac{1}{2}(\mathbf{a}_{1}^{+}\mathbf{a}_{1} - \mathbf{a}_{2}^{+}\mathbf{a}_{2})$. which gives the energy difference between the oscillators, and has the same eigenvalues $-\frac{d}{2}, -\frac{d}{2}+1, \ldots, +\frac{d}{2}$, as a spin system. These operators satisfy the commutation relations:
	\begin{align}\label{eq:cr_schwinger}
	[\mathbf{J}_{+},\mathbf{J}_{-}] = 2\mathbf{J}_{z};\ 
	[\mathbf{J}_{\pm},\mathbf{J}_{z} ] = \pm \mathbf{J}_{\pm}.
	\end{align} 
	The observables of a spin system are equivalent to the Pauli spin matrices \eqref{eq:paulispin}, and satisfy the following:
	\begin{align}
	[\sigma_{x}, \sigma_{y}] = 2i\sigma_{z}.
	\end{align}
	We can define the operators 
	\begin{align}\label{eq:pauli_creation}
	\mathbf{J}_{\pm} = \sigma_{x} \pm i \sigma_{y};\ \mathbf{J}_{z} = \sigma_{z},
	\end{align}
	acting on the states of the spin system. A quick calculation then shows that $\mathbf{J}_{\pm}$ modify the spin component along the $z$-axis by one unit, and the operators \eqref{eq:pauli_creation} satisfy commutation relations identical to \eqref{eq:cr_schwinger}.
\end{proof}

The correspondence between B\'{e}zier curves and spin systems can now be restated using harmonic oscillators. The control points are associated with the energy eigenstates of the oscillator pair, such that for the control point $k$, there are $d-k$ energy quanta in the first oscillator and $k$ quanta in the second, providing another physical interpretation of polar labels. The symmetry property of polar forms is reflected in the fact that only the number of quanta in each oscillator has physical significance. The oscillator model can be considered an alternative formulation of the \emph{P\'{o}lya's urn models} studied by Goldman \cite{goldman1985polya}: the oscillator quanta are directly analogous to the balls put in the urns.

\section{Poisson Curves  and Harmonic Oscillators}\label{sec:poisson_osc}
In this chapter, we outline the relationship between harmonic oscillators and Poisson curves, thus finding further correspondences between physics and CAGD.
\subsection{Poisson Curves and the Analytic Blossom}
Poisson curves are related to analytic functions expressed via a Taylor series, as B\'{e}zier curves are related to polynomials expressed in the power basis \cite{morin2000subdivision,morin2002analytic}. A \emph{Poisson curve} is defined as
\begin{align}
\mathbf{P}(t) = \sum_{i = 0}^{\infty}\mathbf{P}_{i}b_{i}(t),
\end{align} 
where the \emph{Poisson basis functions} 
\begin{align}\label{eq:poisson_basis}
b_{i}(t) = e^{-t}\frac{t^{i}}{i!},\ i \in \mathbb{N},
\end{align}
describe a Poisson probability distribution over the control points $\mathbf{P}_{i}$. These curves share many of their properties with B\'{e}zier curves \cite{morin2002analytic}. In particular, there exists a unique symmetric, multiaffine function $\mathbf{p}(t_{1}, t_{2}, \ldots)$ of infinitely many parameters, called the \emph{analytical blossom}, that enjoys a carefully formulated diagonal property \cite{morin2002analytic,goldman2002affine} and assigns polar labels to the control points of the Poisson curve:
\begin{align}
\mathbf{p}(\underbrace{1, \ldots, 1}_{i}, 0, 0,  \ldots) = \mathbf{P}_{i}.
\end{align}

\subsection{Coherent States for Harmonic Oscillators}
We claim that Poisson curves correspond to harmonic oscillators, mirroring the previously established relationship between B\'{e}zier curves and spin systems. The connection becomes apparent when we consider the coherent states of a harmonic oscillator. 

\begin{theorem}
	The coherent states of a quantum harmonic oscillator define a Poisson probability distribution over the energy eigenstates.
\end{theorem}
\begin{proof}
	A coherent state is a quantum state which saturates the Heisenberg uncertainty relations \cite{gazeau2009coherent}. The simplest example is the vacuum state $\ket{0}$, i.e. when the oscillator contains no energy quanta, described by a Gaussian probability amplitude over the positions \cite{townsend2000modern}. In analogy with what we saw in \autoref{sec:spinfacts:coherent} for spin systems, the coherent states of an oscillator are translated versions of this state in the $x-p$ plane \cite{gazeau2009coherent}. For simplicity, we only consider translations along the position axis. Given the vacuum wavefunction $\ket{0} = \psi(x)$, the version translated by an amount $z$, i.e. $\ket{0}_{+z} := \psi(x-z)$, can be expanded as a Taylor series in $z$:
	\begin{align}\label{eq:taylor}
	\ket{0}_{+z} = \sum_{n = 0}^{\infty}\frac{d^{n}\ket{0}}{dx^{n}}\frac{z^{n}}{n!}.
	\end{align}   
	The derivative is proportional to the momentum operator $\frac{d}{dx} = i\mathbf{p}$, which in turn can be expressed using annihilation and creation operators $\mathbf{p} = \frac{i}{\sqrt{2}}(a - a^{+})$. After substitution of $\mathbf{p}$ into \eqref{eq:taylor} and repeated application of the relations \eqref{eq:creat_effect}, the coherent state can be written as a superposition of energy eigenstates:
	\begin{align}
	\ket{0}_{+z} = e^{-\frac{z^{2}}{2}}\sum_{n = 0}^{\infty}\frac{z^{n}}{\sqrt{n!}}\ket{n},
	\end{align}
	where we made use of the fact $e^{x} = \sum_{n=0}^{\infty}\frac{x^{n}}{n!}$. 
	Squaring the probability amplitudes gives a Poisson distribution over the energy eigenstates, with expectation value $z^{2}$ for the number of energy quanta\footnote{Coherent states of both spin and oscillator systems are parameterized by \emph{squared} coordinates, because we implicitly make use of the \emph{symplectic} moment map (natural for physics), instead of the \emph{algebraic} moment map (natural for CAGD) -- see  \cite{sottile2003toric}.}. 
\end{proof}

\subsection{Poisson Curves from Physics}
We have seen that the coherent states of a quantum oscillator are described by Poisson distributions. It follows then, that the energy eigenstates correspond to control points of a Poisson curve, and the curve itself corresponds to the set of coherent states. The analytical blossom arises naturally as well: each eigenstate is identified by the number of filled, and (an infinite number of) empty energy bins. 

The arguments of \autoref{sec:bezier2physics} could also be employed to arrive at the harmonic oscillator from a Poisson curve. Extending the curve to the complex numbers gives the complex plane $\mathbb{C}$ with a natural symplectic structure and Hamiltonian (energy function), which defines a harmonic oscillator \cite{hall2013quantum}.

We saw earlier that two oscillators containing a fixed number of energy quanta describe a spin system. In terms of coherent states, a pair of Poisson distributed random variables with their sum being fixed results in a binomial distribution, as it is well-known from probability theory. Poisson curves are also described as B\'{e}zier curves of infinite degree, in agreement with the fact that spin-$\frac{d}{2}$ systems degenerate into harmonic oscillators as $d \rightarrow \infty$ \cite[Ch. 7.7]{combescure2012coherent}.

\section{Conclusions and Outlook}\label{sec:conc}
We established a connection between B\'{e}zier and Poisson curves, and physical systems in classical, as well as in quantum mechanics. We note that some of the mathematical and physical concepts relevant to our work have found use recently in geometric design and computer graphics. Geometric quantization have been employed to give representations of vector fields in the context of fluid simulation \cite{weissmann2014smoke,chern2016schrodinger,chern2017inside} -- see \cite{chern2017fluid} for an overview. Complex line bundles were used to optimize direction fields \cite{knoppel2013globally} and stripe patterns \cite{knoppel2015stripe} on surfaces, as well as volumetric deformations \cite{chern2015close} -- see also \cite{knoppel2016complex}. Pauli spin matrices represent (unit) quaternions, and the theory of angular momentum is intimately related to that of the 3D rotation group $SO(3)$ and its double cover $SU(2)$ -- quaternions have long been employed in computer animation \cite{hanson2006visualizing}, but have also found applications in the differential geometry of discrete surfaces \cite{bobenko1994surfaces,kamberov2002quaternions,crane2011spin,crane2013conformal,liu2017dirac,ye2018unified,hoffmann2018discrete,chern2018shape}, as well as in the theory of Pythagorean-Hodograph Curves \cite{choi2002clifford,farouki2008pythagorean}. We also mention that integrable Hamiltonian systems have a well-established relationship with the differential geometry of 3D surfaces \cite{terng2000geometry,rogers2002backlund,bobenko2008discrete}, as well as the dynamics of curves \cite{chou2002integrable,capovilla2002hamiltonians,sato2015generalization,inoguchi2018log,chern2018computing} in the context of infinite-dimensional systems (solitons) \cite{palais1997symmetries}. 

Due to the richness of the related mathematics and physics, there are countless possible avenues for future work -- we sample only the most immediately relevant:
\begin{itemize}
	\item As already mentioned, our approach can also be applied to B\'{e}zier tensor products and simplices. In general, toric varieties correspond to mechanical systems with the special property of \emph{complete integrability}, so that the moment map image is a convex (Newton) polytope \cite{atiyah1983angular,delzant1988hamiltoniens}, and quantization gives the lattice points inside \cite{vergne1996convex,hamilton2007quantization}. This subject is the topic of a planned follow-up paper \cite{vaitkus2018surface}.
	\item Everything we have discussed regarding polar forms, coherent states, geometric quantization, moment maps and oscillators is explained most elegantly using the language of Lie groups and Lie algebras \cite{singer2006linearity,sattinger1986lie,fulton1991representation}, and relates to the symmetric algebra approach of Ramshaw \cite{ramshaw2001paired}, and  combinatorics \cite{blasiak2011combinatorial} as well. We plan to apply these tools to the study of CAGD representations. It appears possible, for example, to derive the $q$-deformed \cite{goldman2015quantum} and umbral \cite{winkel2014generalization} generalizations of B\'{e}zier curves within this framework. It is also an interesting question whether the analogy between P\'{o}lya's urn models and harmonic oscillators carries over to curves based on more complicated urn models.
	\item The control net of B\'{e}zier (and more general toric) curves and surfaces can be characterized as limits of toric degenerations \cite{garc2011toric}. Whether this connection can be exploited in a the context of integrable systems -- following works such as \cite{harada2015integrable} -- remains to be seen. 
	\item The topic of B\'{e}zier curves and surfaces has many other connections to mathematical physics that we have not discussed. Spin systems have an interpretation in terms of the Hall effect for magnetic monopoles \cite{hasebe2013topological}. Toric varieties are known to arise naturally in the context of quantum field theories, known as sigma-models \cite{witten1993phases,hori2003mirror}. Geometric quantization has an alternative interpretation in terms of string theory and branes \cite{gukov2008branes}. The oscillator model for spin systems is an example of a general phenomena wherein certain approximations become exact for some special physical systems \cite{szabo2003equivariant}, which is related to Morse theory in differential topology, and the physics of supersymmetry (Clifford algebras) \cite{witten1982supersymmetry,stone1989supersymmetry} .
	\item Defining B-splines (which are smoothly connected Bernstein polynomials), within a physical framework is the primary open problem for this line of research. Remarkably, splines -- interpreted as volumes of simplex cross-sections \cite{deconcini2010topics,cohen1987cones} -- already play a role in the theory of constrained (reduced) mechanical systems, by a theorem of Duistermaat and Heckman \cite{duistermaat1982variation,atiyah1983angular}. This theorem suggests that arriving at a quantum interpretation of B-splines might be possible using Feynman's path integral formalism \cite{szabo2003equivariant}.   
\end{itemize}

\section*{Acknowledgements}
This project has been supported by the Hungarian Scientific Research Fund (OTKA, No.124727). The author is immensely grateful to his advisor Tam\'{a}s V\'{a}rady for support in pursuing this work. I would also like to thank Ron Goldman, Malcolm Sabin, Kestutis Kar\v{c}iauskas, Alyn Rockwood and G\'{a}bor Etesi for illuminating discussions. 

\section*{References}

\bibliography{QuantumBezier_arXiv.bib}

\end{document}